%% file: TorusBHSreport.tex
\documentclass[a4paper,11pt]{article}
\usepackage{color}
\usepackage{epsfig}
\usepackage{verbatim}
\usepackage{amsmath} 
\usepackage{amssymb}
\usepackage[algoruled,vlined,english,linesnumbered]{algorithm2e} 
\usepackage{enumerate}
\usepackage{xspace}
\usepackage{authblk}
\usepackage{graphicx}
\newcommand{\concat}{\oplus}

\SetAlFnt{\footnotesize}

 \newcommand{\qed}{\hspace*{\fill}\rule{6pt}{6pt}\vspace{.5\smallskipamount}}
 \newtheorem{remark}{Remark}[section]
 \newtheorem{theorem}{Theorem}[section]
 
  \newtheorem{proposition}{Proposition}[section]
 
 \newtheorem{lemma}{Lemma}[section]

 \newtheorem{property}{Property}[section]

 \newenvironment{proof} { \noindent \emph{Proof} : } { \qed }
 \newenvironment{proofof} { \noindent \emph{Proof of}} { \qed }
 \newtheorem{claim}{Claim}[section]
 \newtheorem{case}{Case}[section]

 \newcommand{\cS}{{\cal S}\xspace}
\newcommand{\A}{{\cal A}\xspace}

\newcommand{\InitNR}{{\tt InitNextRing}\xspace}
\newcommand{\NR}{{\tt NextRing}\xspace}
\newcommand{\FR}{{\tt FirstRing}\xspace}
\newcommand{\OneTokenBelow}{{\tt OneTokenBelow}\xspace}
\newcommand{\BlackHoleInNextRing}{{\tt BlackHoleInNextRing}\xspace}
\newcommand{\Analyze}{{\tt Analyze}\xspace}

\begin{document}

\sloppy

\title{Black Hole Search with Finite Automata Scattered in a Synchronous Torus}

\author[1]{J\'{e}r\'{e}mie Chalopin} 
\author[1]{Shantanu Das}
\author[1]{Arnaud Labourel}
\author[2]{Euripides Markou}

\affil[1]{LIF, Aix-Marseille University, Marseille, France. {\small Email:~\{jeremie.chalopin,shantanu.das,arnaud.labourel\}@lif.univ-mrs.fr}}

\affil[2]{Department of Computer Science and Biomedical Informatics,

University of Central Greece, Lamia, Greece. {\small Email:~emarkou@ucg.gr}}

\date{}
\maketitle

\begin{abstract}

We consider the problem of locating a black hole in synchronous anonymous networks using finite state agents. A black hole is a harmful node in the network that destroys any agent visiting that node without leaving any trace. The objective is to locate the black hole without destroying too many agents. This is difficult to achieve when the agents are initially scattered in the network and are unaware of the location of each other.
 In contrast to previous results, we solve the problem using a small team of finite-state agents each carrying a constant number of identical tokens that could be placed on the nodes of the network. Thus, all resources used in our algorithms are independent of the network size. 

We restrict our attention to oriented torus networks and first show that no finite team of finite state agents can solve the problem in such networks, when the tokens are not movable, i.e., they cannot be moved by the agents once they have been released on a node. In case the agents are equipped with movable tokens, we determine lower bounds on the number of agents and tokens required for solving the problem in torus networks of arbitrary size. Further, we present a deterministic solution to the black hole search problem for oriented torus networks, using the minimum number of agents and tokens, thus providing matching upper bounds for the problem.

\bigskip

\end{abstract}

\thispagestyle{empty}
\newpage
\pagenumbering{arabic}


\section{Introduction}

The exploration of an unknown graph by one or more mobile agents is a classical problem initially formulated in 1951 by Shannon \cite{sha51} and it has been extensively studied since then (e.g., see \cite{bs94, dp99, fgkp06}). Recently, the exploration problem has also been studied in unsafe networks which contain malicious hosts of a highly harmful nature, called {\em black holes}. A black hole is a node which contains a stationary process destroying all mobile agents visiting this node, without leaving any trace. 
In the {\em Black Hole Search} problem the goal for the agents is to locate the black hole within finite time. In particular, at least one agent has to survive knowing all edges leading to the black hole. 
The only way of locating a black hole is to have at least one agent visiting it. However, since any agent visiting a black hole is destroyed without leaving any trace, the location of the black hole must be deduced by some communication mechanism employed by the agents. Four such mechanisms have been proposed in the literature: a) the {\em whiteboard} model in which there is a whiteboard at each node of the network where the agents can leave messages, b) the {\em `pure' token} model where the agents carry tokens which they can leave at nodes,  c) the {\em `enhanced' token} model in which the agents can leave tokens at nodes or edges, and d) the time-out mechanism (only for synchronous networks) in which one agent explores a new node while another agent waits for it at a safe node. 

The most powerful inter-agent communication mechanism is having whiteboards at all nodes. Since access to a whiteboard is provided in mutual exclusion, this model could also provide the agents a symmetry-breaking mechanism: If the agents start at the same node, they can get distinct identities and then the distinct agents can assign different labels to all nodes. Hence in this model, if the agents are initially co-located, both the agents and the nodes can be assumed to be non-anonymous without any loss of generality. 

In asynchronous networks and given that all agents initially start at the same safe node, the Black Hole Search (BHS) problem has been studied under the whiteboard model (e.g.,~\cite{dfkprs06,dfps06,dfps07,dfs04}), the `enhanced' token model (e.g.,~\cite{dfks06,dkss06,shi09}) and the `pure' token model in \cite{fis08}.  It has been proved that the problem can be solved with a minimal number of agents performing a polynomial number of moves. Notice that in an asynchronous network the number of the nodes of the network must be known to the agents otherwise the problem is unsolvable (\cite{dfps07}). If the graph topology is unknown, at least $\Delta +1$ agents are needed, where $\Delta$ is the maximum node degree in the graph (\cite{dfps06}). Furthermore the network should be $2$-connected. It is also not possible to answer the question of {\em whether} a black hole exists in the network.

In asynchronous networks, with scattered agents (not initially located at the same node), the problem has been investigated for the ring topology (\cite{dfps03,dfps07}) and for arbitrary topologies (\cite{fkms09,cds07}) in the whiteboard model while in the `enhanced' token model it has been studied for rings (\cite{dss07,dss08}) and for some interconnected networks (\cite{shi09}).  

The consideration of synchronous networks makes a dramatic change to the problem. Now two co-located distinct agents can discover one black hole in any graph by using the time-out mechanism, without the need of whiteboards or tokens. Moreover, it is now possible to answer the question of whether a black hole actually exists or not in the network. No knowledge about the number of nodes is needed. Hence, with co-located distinct agents, the issue is not the feasibility but the time efficiency of black hole search. The issue of efficient black hole search has been studied in synchronous networks without whiteboards or tokens (only using the time-out mechanism) in \cite{ckr06,ckr10,ckmp06,ckmp07,kmrs07,kmrs08,knp09} under the condition that all distinct agents start at the same node. However when the agents are scattered in the network, the time-out mechanism is not sufficient anymore.

Indeed the problem seems to be much more difficult in the case of scattered agents and there are very few known results for this scenario. In this paper we study this version of the problem using very simple agents that can be modeled as finite state automata. Our objective is to determine the minimum resources, such as number of agents and tokens, necessary and sufficient to solve the problem in a given class of networks. For the class of ring networks, recent results \cite{cdlm11} show that having constant-size memory is not a limitation for the agents when solving this problem. We consider here the more challenging scenario of anonymous torus networks of arbitrary size. 
We show that even in this case, finite state agents are capable of locating the black hole in all oriented torus networks using only a few tokens. Note that the exploration of anonymous oriented torus networks is a challenging problem in itself, in the presence of multiple identical agents~\cite{kkm10}. Since the tokens used by the agents are identical, an agent cannot distinguish its tokens from those of another agent.

While the token model has been mostly used in the exploration of safe networks, the whiteboard model is commonly used in unsafe networks. The `pure' token model can be implemented with $O(1)$-bit whiteboards, for a constant number of agents and a constant number of tokens, while the `enhanced' token model can be implemented having a $O(\log d)$-bit whiteboard on a node with degree $d$. In the whiteboard model, the capacity of each whiteboard is always assumed to be of at least $\Omega (\log n)$ bits, where $n$ is the number of nodes of the network. 
In all previous papers studying the Black Hole Search problem under a token model apart from \cite{fis08} and \cite{cdlm11}, the authors have used the `enhanced' token model with agents having non-constant memory. The weakest `pure' token model has been used in \cite{fis08} for co-located non-constant memory agents equipped with  a map in asynchronous networks.

The Black Hole Search problem has also been studied for co-located agents in asynchronous and synchronous directed graphs with whiteboards in \cite{cdkmp09,knp09}. In \cite{gla09} they study the problem in asynchronous networks with whiteboards and co-located agents without the knowledge of incoming link. A different dangerous behavior is studied for co-located agents in \cite{km10}, where the authors consider a ring and assume black holes with Byzantine behavior, which do not always destroy a visiting agent.


\bigskip

\noindent {\bf Our Contributions:}
We consider the problem of locating the black hole in an anonymous but oriented torus containing exactly one black hole, using a team of identical agents that are initially scattered within the torus. Henceforth we will refer to this problem as the BHS problem. 
We focus our attention on very simple mobile agents. The agents have constant-size memory, they can communicate with other agents only when they meet at the same node and they carry a constant number of identical tokens which can be placed at nodes. 
The tokens may be movable (i.e. they can be released and picked up later) or unmovable (i.e. they cannot be moved by the agents once they have been released on a node).
%
We prove the following results: 
\begin{itemize}
\item No finite team of agents can solve the BHS problem in all oriented torus networks using a finite number of \emph{unmovable} tokens.

\item For agents carrying any finite number of \emph{movable} tokens, at least three agents are required to solve the problem.

\item Any algorithm for solving BHS using $3$ agents requires more than one movable token per agent. 

\item The BHS problem can be solved using three agents and only two movable tokens per agent, thus matching both the lower bounds mentioned above.
\end{itemize}

In Section~\ref{sec:model}, we formally describe our model, giving the capabilities of the agents.
In Section~\ref{sec:impos}, we prove lower bound on the number of agents and tokens needed to solve the BHS problem in the torus.
In Section~\ref{sc:simple_algo}, we present two deterministic algorithms for BHS: (i) using $k\geq3$ agents carrying 3 movable tokens per agents, and (ii) using $k\geq4$ agents carrying 2 movable tokens per agent. In Section~\ref{sec:difficultAlgo}, we present a more involved algorithm that uses exactly 3 agents and 2 tokens per agent thus meeting the lower bounds. All our algorithms are time-optimal and since they do not require any knowledge about the dimensions of the torus, they work in any synchronous oriented torus, using only a finite number of agents having constant-size memory.

\section{Our Model}\label{sec:model}

Our model consists of $k \geq 2$ anonymous and identical mobile agents that are initially placed at distinct nodes 
of an anonymous, synchronous torus network of size $n\times m$, $n\geq 3$, $m\geq 3$.  We assume that the torus is oriented, i.e., at each node, the four incident edges are consistently marked as North, East, South and West. Each mobile agent owns a constant number of $t$ identical tokens which can be placed at any node visited by the agent. In all our protocols a node may contain at most three tokens at the same time and an agent carries at most three tokens at any time. A token or an agent at a given node is visible to all agents on the same node, but is not visible to any other agent. The agents follow the same deterministic algorithm and begin execution at the same time and being at the same initial state.

At any single time unit, a mobile agent occupies a node $u$ of the network and may 1) detect the presence of one or more tokens and/or agents at node $u$, 2) release/take one or more tokens to/from the node $u$, and 3) decide to stay at the same node or move to an adjacent node. We call a token {\em movable} if it can be moved by any mobile agent to any node of the network, otherwise we call the token {\em unmovable} in the sense that, once released, it can occupy only the node in which it has been released.

Formally we consider a mobile agent as a finite Moore automaton 
$\A=(\cS,S_0,\Sigma,\Lambda,\delta,\phi)$, where $\cS$ is a set of $\sigma \geq 2$ states; $S_0$ is the {\em initial} state; $\Sigma$ is the set of possible configurations an agent can see when it enters a node; $\delta: \cS\times \Sigma\rightarrow\cS$ is the transition function; and $\phi:\cS\rightarrow\Lambda$ is the output function.
Elements of $\Sigma$ are quadruplets $(D,x,y,b)$ where $D \in \{\tt North, South, East, West, none\}$ is the direction through which the agent has arrived at the node, $x$ is the number of tokens (at most $3$) at that node, $y$ is number of tokens (at most $3$) carried by the agent and $b\in\{true, false\}$ indicates whether there is at least another agent at the node or not.
Elements of $\Lambda$ are quadruplets $(P,s,X,M)$ where  $P\in\{\tt put, pick\}$ is the action performed by the agent on the tokens, $s\in\{0,1,2,3\}$ is the number of tokens concerned by the action $A$, $X \in \{\tt North, South, East, West, none\}$ is the edge marked as dangerous by the agent, and $M \in \{\tt North, South, East, West, none\}$ is the move performed by the agent. Note that the agent always performs the action before the marking and the marking before the move.

Note that all computations by the agents are independent of the size $n\times m$ of the network since the agents have no knowledge of $n$ or $m$.
There is exactly one black hole in the network. An agent can start from any node other than the black hole and no two agents are initially co-located.  Once an agent detects a link to the black hole, it marks the link permanently as dangerous (i.e., disables this link). Since the agents do not have enough memory to remember the location of the black hole, we require that at the end of a black hole search scheme, all links incident to the black hole (and only those links) are marked dangerous and that there is at least one surviving agent. Thus, our definition of a successful BHS scheme is slightly different from the original definition.
%
The time complexity of a BHS scheme is the number of time units needed for completion of the scheme, assuming the worst-case location of the black hole and the worst-case initial placement of the scattered agents.



\section{Impossibility results}\label{sec:impos}

In this section we give lower bounds on the number of agents and the number and type of tokens needed for solving the BHS problem in any anonymous, synchronous and oriented torus.

\subsection{Agents with unmovable tokens}

We will prove that any constant number of agents carrying a constant number of unmovable tokens each, can not solve the BHS problem in an oriented torus. The idea of the proof is the following: We show that an adversary (by looking at the transition function of an agent) can always select a big enough torus and initially place the agents so that no agent visits nodes which contain tokens left by another agent, or meets with another agent. Moreover there are nodes on the torus never visited by any agent. Hence the adversary may place the black hole at a node not visited by any of the agents to make the algorithm fail. 



\begin{theorem}\label{impos-torus-k-k}
For any constant numbers $k,t$, there exists no algorithm that solves BHS in all oriented tori containing one black hole and $k$ scattered agents, where each agent has a constant memory and $t$ unmovable tokens.
\end{theorem}

To prove Theorem \ref{impos-torus-k-k} we will need the following two propositions which appeared in \cite{kkm10}.

\begin{proposition}\cite{kkm10}\label{pr-manyuntokens}
Consider one mobile agent with $\sigma$ states and a constant number $t$ of identical unmovable tokens. We can always (for any configuration of the automaton, i.e., states and transition function) select a $n \times n$ oriented torus, where $n >t \sigma^2$ so that no matter what is the starting position of the agent, it cannot visit all nodes of the torus. In fact, the agent will visit at most $\sigma + t(\sigma - 1)^2 (n+1) < n^2$ nodes.
\end{proposition}

\begin{proposition}\cite{kkm10}\label{pr-lemma-nodev}
Let $A$ be an agent with $\sigma$ states and a constant number $t$ of identical unmovable tokens in a $n \times n$ oriented torus, where $n > t \sigma^2$ and let $v$ be a node in that torus. There are at most $\sigma + t(\sigma - 1)^2 (n+1) < n^2$ different starting nodes that we could have initially placed $A$ so that node $v$ is always visited by $A$.
\end{proposition}

\begin{proofof}\emph{ Theorem \ref{impos-torus-k-k}:}
Consider a constant number of $k$ mobile agents with $\sigma$ states and a constant number of $t$ unmovable tokens each, in an $n \times n$ oriented torus. We show that an adversary can always (for any configuration of the automatons, i.e., states and transition function) initially place the agents on the torus and select $n$ so that there are nodes on the torus never visited by any agent.

Take a $n \times n$ oriented torus, where $n > 2k t^2\sigma^2$ and let $s(A_1)$ be the starting node of agent $A_1$. If agent $A_1$ was alone in the torus would release its tokens at nodes $T_1(A_1), T_2(A_1), \ldots, T_t(A_1)$. According to Proposition \ref{pr-lemma-nodev}, there are at least $n^2 - (\sigma + t(\sigma - 1)^2 (n+1))$ starting nodes at which an adversary could place agent $A_2$ so that $A_2$ does not visit node $T_1(A_1)$. Among these starting nodes (applying again Proposition \ref{pr-lemma-nodev}) there are at most $\sigma + t(\sigma - 1)^2 (n+1)$ nodes that would lead agent $A_2$ to token $T_2(A_1)$, another at most $\sigma + t(\sigma - 1)^2 (n+1)$ nodes that would lead agent $A_2$ to token $T_3(A_1)$ and so on. Therefore there are at least $n^2 - t(\sigma + t(\sigma - 1)^2 (n+1))$ starting nodes at which the adversary could place agent $A_2$ so that $A_2$ does not visit any of the $T_1(A_1), T_2(A_1), \ldots, T_t(A_1)$ nodes. The adversary still needs to place agent $A_2$ at a starting node $s(A_2)$ so that $A_2$ releases its tokens at nodes $T_1(A_2), T_2(A_2), \ldots, T_t(A_2)$ not visited by agent $A_1$. 

Notice that an agent can decide to release a new token at a distance of at most $\sigma$ nodes from a previously released token (an agent cannot count more than $\sigma$ before it repeats a state). Since $n > 2k t^2\sigma^2$ for every two different starting nodes $s(A_2)$ and $s'(A_2)$, agent $A_2$ would release its tokens to nodes $T_1(A_2), T_2(A_2), \ldots, T_t(A_2)$ and $T'_1(A_2), T'_2(A_2), \ldots, T'_t(A_2)$ respectively, where $T_i(A_2) \neq T'_i(A_2), 1 \leq i \leq t$.

Since in view of Proposition \ref{pr-manyuntokens} agent $A_1$ can visit at most $\sigma + t(\sigma - 1)^2 (n+1)$ nodes (if $A_1$ was alone in the torus), there are at most $\sigma + t(\sigma - 1)^2 (n+1)$ starting nodes for $A_2$ for which $A_2$ would place its first token at a node visited by agent $A_1$, another at most $\sigma + t(\sigma - 1)^2 (n+1)$ starting nodes for $A_2$ for which $A_2$ would place its second token at a node visited by agent $A_1$, and so on. Hence we need to exclude another $t(\sigma + t(\sigma - 1)^2 (n+1))$ starting nodes for agent $A_2$. Thus we have left with $n^2 - 2t(\sigma + t(\sigma - 1)^2 (n+1))$ starting nodes at which the adversary can place agent $A_2$ so that $A_2$ does not visit any of the $T_1(A_1), T_2(A_1), \ldots, T_t(A_1)$ nodes and agent $A_1$ does not visit any of the $T_1(A_2), T_2(A_2), \ldots, T_t(A_2)$ nodes. 

For the placement of agent $A_3$, following the same reasoning, and using again Proposition \ref{pr-lemma-nodev}, we have that there are at least $n^2 - 2t(\sigma + t(\sigma - 1)^2 (n+1))$ starting nodes at which the adversary could place agent $A_3$ so that $A_3$ does not visit any of the $T_1(A_1), T_2(A_1), \ldots, T_t(A_1)$ or $T_1(A_2), T_2(A_2), \ldots, T_t(A_2)$ nodes. And using Proposition \ref{pr-manyuntokens} as above by excluding another $2t(\sigma + t(\sigma - 1)^2 (n+1))$ starting nodes for agent $A_3$, we have left with $n^2 - 4t(\sigma + t(\sigma - 1)^2 (n+1))$ starting nodes at which the adversary can place agent $A_3$ so that $A_3$ does not visit any of the $T_1(A_1), T_2(A_1), \ldots, T_t(A_1)$ or $T_1(A_2), T_2(A_2), \ldots, T_t(A_2)$ nodes and agents $A_1$ and $A_2$ do not visit any of the $T_1(A_3), T_2(A_3), \ldots, T_t(A_3)$ nodes.

Following the same reasoning, an adversary can select a node out of 
\[n^2 - 2(k -1)t(\sigma + t(\sigma - 1)^2 (n+1))\] 
nodes to place agent $A_{k}$ so that $A_{k}$ does not visit any of the $T_1(A_j), T_2(A_j), \ldots, T_t(A_j)$ nodes, where $1 \leq j \leq k-1$, and agents $A_j$ do not visit any of the $T_1(A_{k}), T_2(A_{k}), \ldots, T_t(A_{k})$ nodes.

Hence all agents may only visit their own tokens and they have to do it at the same time and being at the same states and therefore they maintain their initial distance forever. Since any agent may only see its own tokens, by Proposition \ref{pr-manyuntokens}, any agent can visit at most $\sigma + t(\sigma - 1)^2 (n+1)$ nodes on the torus and hence all agents will visit at most $k(\sigma + t(\sigma - 1)^2 (n+1)) < n^2$ nodes when $n > 2k t^2\sigma^2$. Therefore there are nodes never visited by any agent and the adversary can place the black hole at such a node.
\end{proofof}

\subsection{Agents with movable tokens}

We first show that the BHS problem is unsolvable in any synchronous torus by two scattered agents having any number of movable tokens even if the agents have unlimited memory.

\begin{lemma}\label{impos-torus-2-k}
Two agents carrying any number of movable tokens cannot solve the BHS problem in an oriented torus even if the agents have unlimited memory.
\end{lemma}

\begin{proof} 
 Assume w.l.o.g. that the first move of the agents is going East. Suppose that the black hole has been placed by an adversary at the East neighbor of an agent. This agent vanishes into the black hole after its first move. The adversary places the second agent such that it vanishes into the black hole after its first vertical move, or it is in a horizontal ring not containing the black hole if the agent never performs vertical moves. Observe that the trajectories of the two agents intersect only at the black hole and neither can see any token left by the other agent. Neither of the agents will ever visit the East neighbor of the black hole and thus, they will not be able to correctly mark all links incident to the black hole.
\end{proof}

Thus, at least three agents are needed to solve the problem. We now determine a lower bound on the number of tokens needed by three scattered agents to solve BHS.

\begin{lemma}\label{impos-torus-3-2}
There exists no universal algorithm that could solve the BHS problem in all oriented tori using three agents with constant memory and one movable token each.
\end{lemma}

\begin{proof} 
Clearly, in view of Theorem \ref{impos-torus-k-k}, an algorithm which does not instruct an agent to leave its token at a node, cannot solve the BHS problem. Hence any potentially correct algorithm should instruct an agent to leave its token down. Moreover this decision has to be taken after a finite number of steps (due to agents' constant memory). After that the agents visit new nodes until they see a token. Following a similar reasoning as in Theorem \ref{impos-torus-k-k} we can show that if the agents visit only a constant number of nodes before returning to meet their tokens they cannot visit all nodes of the torus. If they move their tokens each time they see them and repeat the previous procedure (i.e., visit a constant number of nodes and return to meet their tokens), we can show that they will find themselves back at their initial locations and initial states without having met with other agents and leaving some nodes unvisited. An adversary may place the black hole at an unvisited node to make the algorithm fail. Now consider the case that at some point an algorithm instructs the agents to visit a non-constant number of nodes until they see a token (e.g., leave your token down and go east until you see a token). Again in a similar reasoning as in Theorem \ref{impos-torus-k-k}, we can show that an adversary may initially place the agents and the black hole, and select the size of the torus so that two of the agents enter the black hole without leaving their tokens close to it: The agent (say $A$) that enters first into the black hole has been initially placed by an adversary so that it left its token more than a constant number of nodes away from the black hole. The adversary initially places another agent $B$ so that 
it enters the black hole before it meets $A$'s token. Furthermore $B$ leaves its token more than a constant number of nodes away from the black hole. Hence the third agent, even if it meets the tokens left by $A$ or $B$, it could not decide the exact location of the black hole.
\end{proof}

We can prove using Lemma \ref{impos-torus-3-2} and \ref{impos-torus-2-k}, the following theorem.

\begin{theorem}\label{th:imp}
At least three agents are necessary to solve the BHS problem in an oriented torus of arbitrary size.
Any algorithm solving this problem using three agents requires at least two movable tokens per agent.
\end{theorem}

\section{Simple Algorithms for BHS in a torus using moveable tokens}\label{sc:simple_algo}

Due to the impossibility result from the previous section, we know that unmoveable tokens are not sufficient to solve BHS in a torus. In the following, we will use only moveable tokens.
To explore the torus an agent uses the Cautious-Walk technique~\cite{dfps07} using moveable tokens. In our algorithms, a \emph{Cautious-Walk in direction $D$ with $x$ tokens} means the following: (i) the agent releases a sufficient number of tokens such that there are exactly $x$ tokens at the current node, (ii) the agent moves one step in direction $D$ and if it survives, the agent immediately returns to the previous node, (iii) the agent picks up the tokens released in step (i) and again goes one step in direction $D$. If an agent vanishes during step (ii), any other agent arriving at the same location sees $x$ tokens and realizes a potential danger in direction $D$. Depending on the algorithm an agent may use 1, 2, or 3 tokens to implement the Cautious-Walk.

\subsection{Solving BHS using $k\geq 3$ agents and 3 tokens}

We show that three agents suffice to locate the black hole if the agents are provided with three tokens. We present an algorithm (BHS-torus-33) that achieves this. This algorithm uses three procedures : \emph{Mark\_All}, \emph{Wait} and \emph{Cautious-walk}.

\noindent The procedure \emph{Mark\_All}($D$) is executed by an agent when it deduces that the next node $w$ in the direction $D$ from current location is the black hole. The agent traverses a cycle around node $w$, visiting each neighbor of $w$ and marking as ``dangerous" all links leading to node $w$.

\noindent The subroutine \emph{Wait}($x$) requires the agent to take no action for $x$ time units.

\begin{procedure}
 \caption{Cautious-walk(Direction $D$, integer $x$)}
 \tcc{Procedure used by the agent to explore the nodes}
\SetKw{go}{Go}

Put tokens until there are $x$ tokens\;
\go one step along $D$ and then go back\tcc*{test if the node in direction $D$ is the black hole}
Pick up the tokens released in the first step\;
\go one step along $D$\;
\end{procedure}

\begin{algorithm}
\tcc{ Algorithm for BHS in Oriented Torus (using k=3 or more  agents, 3 tokens each)}
\SetKwFunction{WAIT}{Wait}
\SetKwFunction{MarkAll}{Mark-All}
\SetKwFunction{CautiousWalk}{CautiousWalk}
\SetKw{go}{Go}
\SetKw{and}{and}
\SetKw{ore}{or}
\tcc{ One token = Homebase \\
Two tokens = BlackHole in the East \\
Three tokens = BlackHole in the South}
~\\

Found:= false\;
\While{not Found}{
Count := 0\;
Put two tokens\; 
Go one step East and come back\;
Pick one token and go one step East \tcc*{leaving one token at the homebase}
 \Repeat{found two tokens \ore $count=2$}{
   \lIf{found single token}{increment Count\;}
   \CautiousWalk{East,2}\;
 }
 \lIf {found 2 or 3 tokens} { Found:=true\; }
 \Else(\tcc*[f]{The agent found 1 token and must move to the next horizontal ring})
 {
	Pick one token \tcc*{Remove the homebase token}
	\CautiousWalk{South,3} \tcc*{using the token at the current node}
 	  \While(\tcc*[f]{Current node is the homebase of another agent}){found one token}{
 	  	\CautiousWalk{East,2} \tcc*{using the token at the current node} 
	}
   \lIf {found 2 or 3 tokens} { Found:=true; }
 }
}
\lIf {found 2 tokens} { \MarkAll{East}\; }
\lIf {found 3 tokens} { \MarkAll{South}\; }

 \caption{BHS-Torus-33}
  \label{algo:BHS-T-3-3}
\end{algorithm}

\noindent \textbf{Algorithm BHS-torus-33}:\\ 
\noindent An agent explores one horizontal ring at a time and then moves one step South to the next horizontal ring and so on. 
When exploring a horizontal ring, the agent leaves one token on the starting node. This node is called the \emph{homebase} of the agent and the token left (called homebase token) is used by the agent to decide when to proceed to the next ring. The agent then uses the two remaining tokens to repeat Cautious-Walk in the East direction until it has seen twice a node containing one token. Any node containing one token is a homebase either of this agent or of another agent. The agent moves to the next horizontal ring below after encountering two homebases. However before moving to the next ring, it does a cautious walk in the South direction with three tokens (the two tokens it carries plus the homebase token). If the agent survives and the node reached by the agent has one token, the agent repeats a cautious walk in the East direction (with two tokens) until it reaches an empty node. The agent can now use this empty node as its new homebase. It then repeats the same exploration process for this new ring leaving one token at its new homebase. 

Whenever the agent sees two or three tokens at the end of a cautious-walk, the agent has detected the location of the black hole: If there are two (resp. three) tokens at the current node, the black hole is the neighboring node $w$ to the East (resp. South). In this case, the agent stops its normal execution and then traverses a cycle around node $w$, visiting all neighbors of $w$ and marking all the links leading to $w$ as dangerous.

\begin{theorem}
Algorithm BHS-torus-33 correctly solves the BHS problem with $3$ or more agents.
\end{theorem}

\begin{proof}
An agent may visit an unexplored node only while going East or South. If one agent enters the black hole going East (resp. South), there will be two (resp. three) tokens on the previous node and thus, no other agent would enter the black hole through the same link. This implies that at least one agent always survives. 
Whenever an agent encounters two or three tokens at the end of a Cautious-Walk, the agent is certain of the location of the black hole since any alive agent would have picked up its Cautious-Walk tokens in the second step of the cautious walk (The agents move synchronously always using cautious walk and taking three steps for each move). 
\end{proof}
\subsection{BHS using $k\geq4$ agents and 2 tokens each}
We now present an algorithm that uses only two tokens per agent, but requires at least $4$ agents to solve BHS.

%
During the algorithm, the agents put two tokens on a node $u$ to signal that either the black hole is on the South or the East of node $u$. Eventually, both the North neighbor and the West neighbor of the black hole are marked with two tokens.  Whenever there is a node $v$ such that there are exactly two tokens at both the West neighbor of $v$ and the North neighbor of $v$, then we say that there exists a \emph{Black-Hole-Configuration} (BHC) at $v$.

\bigskip

\noindent \textbf{Algorithm BHS-torus-42}: 

\noindent The agent puts two tokens on its starting node (homebase). It then performs a Cautious-Walk in the East direction. If the agent survives, it returns, picks up one token and repeats the Cautious-Walk with one token in the East direction (leaving the other token on the homebase) until it reaches a node containing one or two tokens. 

\begin{itemize}
\vspace{-0.3cm}
\item If the agent reaches a node $u$ containing two tokens, then the black hole is the next node on the East or on the South (See Property C of Proposition~\ref{prop:alg42}). The agent stops its exploration and checks whether the black hole is the East neighbor or the South neighbor. 
\vspace{-0.3cm}
\item Whenever an agent reaches a node containing one token, it performs a Cautious-Walk in East direction with two tokens and then continues the Cautious-Walk in the same direction with one token. If the agent encounters three times\footnote{The agent may encounter homebases of two agents which have both fallen into the black hole. (In this case it must continue in the same direction until it locates the black hole)} a node containing one token, it moves to the next horizontal ring below. To do that it first performs a Cautious-Walk with two tokens in the South direction. If the agent survives and reaches the ring below, it can start exploring this horizontal ring. If the current node is not empty, the agent repeats a cautious walk with two token in the East direction until it reaches an empty node. Now the agent repeats the same exploration process using the empty node as its new homebase.
Whenever the agent encounter a node with two tokens, it stops its exploration and checks whether the black hole is the East or South neighbor.

\end{itemize}

\begin{algorithm}
\tcc{ Algorithm for BHS in Oriented Torus (using k=4 or more  agents, 2 tokens) }
\SetKw{and}{and}
\SetKwFunction{WAIT}{Wait}
\SetKwFunction{MarkAll}{Mark-All}
\SetKwFunction{CautiousWalk}{CautiousWalk}
\SetKw{go}{Go}
\SetKw{ore}{or}
\tcc{ One token = Homebase (or Blackhole in the East) \\
Two tokens = BlackHole either in the East or in the South}
~\\

Found:= false\;
\While{not Found}{
 Count := 0\;
 Put one token \tcc*{mark your homebase}
 \CautiousWalk{East,2}\;
 \Repeat{found two tokens \ore $count=3$}{
   \uIf{found single token}{increment Count; \CautiousWalk{East,2};}
   \ElseIf {found no tokens} { \CautiousWalk{East,1}\;}
 }
 
 \uIf{found 2 tokens} { 
   Found:=true\;
  }
 \Else(\tcc*[f]{found 1 token (thrice), so move to the next horizontal ring})
   {
	Pick one token \tcc*{Remove the homebase token}
	\CautiousWalk{South,2}\;
 	  \While(\tcc*[f]{Search for an empty node}){found one token}
	  {
 	  	\CautiousWalk{East,2}\;
 	  }
   	\lIf {found 2 tokens} { 
	Found:=true\;
        }
   }
}
\tcc{The agent found two tokens and knows that one of the neighbors is the black hole}
\go West\;
Wait(1) \tcc*{To Synchronize with Cautious-Walk}
\go South\;

\lIf {found two tokens} { \MarkAll{East}\; }
\Else{
 \uIf {found one token} { \CautiousWalk{East,2}\; }
 \Else { \CautiousWalk{East,1}\; }
 \go North\;
 \MarkAll{East}\;
}

 \caption{BHS-Torus-42}
  \label{algo:BHS-T-4-2}
\end{algorithm}

\noindent In order to check if the black hole is the East neighbor $v$ or the South neighbor $w$ of the current node $u$ (containing two tokens), the agent performs the following actions: The agent reaches the West neighbor $x$ of $w$ in exactly three time units by going west and south (and waiting one step in between). If there are two tokens on this node $x$ then $w$ is the black hole. Otherwise, the agent 
performs a cautious walk in the East direction with one token (or with two tokens if there is already one token on node $x$). If it safely reaches node $w$, then the black hole is the other node $v$. Otherwise the agent would have fallen into the black hole leaving a BHC at node $w$. 
An agent that discovers the black hole, traverses a cycle around the black hole visiting all its neighbors and marking all the links leading to it as dangerous.

\begin{proposition}\label{prop:alg42}
During an execution of BHS-torus-42 with $k \geq 4$ agents, the following properties hold:
\begin{enumerate}[A]
\vspace{-0.25cm}
\item When an agent checks the number of tokens at a node, all surviving agents have picked up their cautious-walk token.

\vspace{-0.25cm}\item At most three agents can enter the black hole:
\begin{enumerate}\vspace{-0.2cm}
\item at most one agent going South leaving two tokens at the previous node.
\item at most two agents going East, each leaving one of its tokens at the previous node.
\end{enumerate}
\vspace{-0.25cm}
\item When an agent checks the number of tokens at a node, if there are two tokens then the black hole is either the East or the South neighbor of the current node.
\vspace{-0.25cm}
\item After an agent starts exploring a horizontal ring, one of the following eventually occurs:
\begin{enumerate}\vspace{-0.2cm}
\item If this ring is safe, the agent eventually moves to the next horizontal ring below.
\item Otherwise, either all agents on this ring fall into the black hole or one of these agents marks all links to the black hole.
\end{enumerate}
\end{enumerate}
\end{proposition}

\begin{theorem}
Algorithm BHS-torus-42 correctly solves the black hole search problem with $k \geq 4$ agents, each having two tokens.
\end{theorem}

\begin{proof}
Property $B$ of Proposition~\ref{prop:alg42} guarantees that at least one agent never enters the black hole. Property $D$ ensures that one of the surviving agents will identify the black hole. Property $C$ shows that if the links incident to a node $w$ are marked as dangerous by the algorithm, then node $w$ is the black hole. 
\end{proof}

\newcommand{\CWtogether}{\emph{Cautious-Walk-With-Another-Agent}}
\newcommand{\BHStogether}{\emph{BHS-with-colocated-agents}}

\section{Optimal algorithm for BHS using 3 agents and 2 tokens each}\label{sec:difficultAlgo}
\subsection{Sketch of the algorithm}

Using the techniques presented so far, we now present the algorithm that uses exactly 3 agents and two tokens per agent.
The algorithm is quite involved and we present here only the main ideas of the algorithm. The complete algorithm along with a proof of correctness can be found in Subsection~\ref{sec-app-algo} and \ref{sec-app-proof}.

Notice first that we can not prevent two of the three agents from
falling into the black hole (see proof of
Theorem~\ref{impos-torus-2-k}). To ensure that no more than two agents
enters the black hole, the algorithm should allow the agent to move
only in two of the possible four directions (when visiting unexplored
nodes). When exploring the first horizontal ring, an agent always
moves in the East direction, using a Cautious-Walk as before and
keeping one token on the starting node (homebase). This is called
procedure \emph{First-Ring}. Once an agent has completely explored one
horizontal ring, it explores the ring below, using procedure
\emph{Next-Ring}. During procedure Next-Ring, an agent traverses the
already explored ring and at each node $u$ of this ring, the agent
puts one token, traverses one edge South (to check the node just below
node $u$), and then immediately returns to node $u$ and picks up the
token.  Note that an agent may fall into the black hole only when
going South during procedure \emph{Next-Ring} or when going East
during procedure \emph{First-Ring}. We ensure that at most one agent
falls into the black hole from each of these two directions. The
surviving agent must then identify the black hole from the tokens left
by the other agents.
%
%
For this algorithm, we redefine the \emph{Black-Hole-Configuration}
(BHC) as follows: If there is a node $v$ such that there is one or two
tokens at both the West neighbor of $v$ and the North neighbor of $v$,
then we say that a BHC exists at $v$. The algorithm should avoid
forming a black hole configuration at any other node except the
black hole. In particular, when the agents put tokens on their
homebase, these tokens should not form a BHC. This requires
coordination between any two agents that are operating close to
each other (e.g. in adjacent rings of the torus) and it is not always
possible to ensure that a BHC is never formed at a safe node. 

The main idea of the algorithm is to make two agents meet whenever
they are close to each other (this requires a complicated
synchronization and checking procedure). If any two agents manage to
gather at a node, we can easily solve BHS using the standard procedure
for colocated agents\footnote{Note that the agents meeting at a node
  can be assigned distinct identifiers since they would arrive from
  different directions.} with the time-out mechanism (see \cite{cdlm11,ckr06}) .  On the other hand, if the agents are always far away from
each other (i.e. more than a constant distance) then they do not
interfere with the operations of each other until one of them falls
into the black hole. The agents explore each ring, other than the first
ring, using procedure \emph{Next-Ring}.  We have another procedure
called \emph{Init-Next-Ring} that is always executed at the beginning
of \emph{Next-Ring}, where the agents check for certain special
configurations and take appropriate action. If the tokens on the potential homebases of two agents would form a BHC on a safe
node, then we ensure two agents meet.

\bigskip

\noindent \textbf{Synchronization:} 

\noindent During the algorithm, we ensure that two agents meet if they start the procedure \emph{Init-Next-Ring} from nearby locations. We achieve this by keeping the agents synchronized to each other, in the sense that they start executing the procedure at the same time, in each iteration. More precisely, we ensure the following property:

\begin{property}\label{pr:synch}
When one agent starts procedure \emph{Init-Next-Ring}, any other
surviving agent either (i) starts procedure \emph{Init-Next-Ring} at
exactly the same time, or (ii) waits in its current homebase along
with both its tokens during the time the other agent is executing the
procedure or, (iii) has not placed any tokens at its homebase. 
\end{property}

Notice that if there are more than one agent initially located at distinct nodes within the same horizontal ring, an agent cannot distinguish its homebase from the homebase of another agent, and thus an agent would not know when to stop traversing this ring and go down to the next one. We get around this problem by making each agent traverse the ring a sufficient number of times to ensure that every node in this ring is explored at least once by this agent. To be more precise, each agent will traverse the ring until it has encountered a homebase node six times (recall that there can be either one, two or three agents on the same ring). This implies that in reality the agent may traverse the same ring either twice, or thrice, or six times. 
If either all agents start in distinct rings or if all start in the same ring then, the agents would always be synchronized with each other (i.e. each agent would start the next iteration of \emph{Next-Ring} at the same time).
The only problem occurs when two agents start on the same ring and the third agent is on a different ring. In this case, the first two agents will be twice as fast as the third agent. We introduce waiting periods appropriately to ensure that Property~\ref{pr:synch} holds.

For both the procedures \emph{First-Ring} and \emph{Next-Ring}, we
define one \emph{big-step} to be the period between when an agent
arrives at a node $v$ from the West with its token(s) and when it has
left node $v$ to the East with its token(s). During a big-step the agent
would move to an unsafe node (East or South), come back to pick its
token, wait for some time at $v$ before moving to the next node with
its token. The waiting period is adjusted so that an agent can execute
the whole procedure \emph{Init-Next-Ring} during this time.  Thus, the
actual number of time units for each big-step is a constant $D$ which we
call the \emph{magic number}.

\bigskip

\noindent \textbf{Algorithm \emph{BHS-Torus-32}:}\\

\noindent \textbf{Procedure} \emph{First-Ring}:\\
\noindent During this procedure the agent explores the horizontal ring
that contains its starting location.  The agent puts one token on its
homebase and uses the other token to perform cautious-walk in the
direction East, until it enters a node with some tokens. If it finds a
node with two tokens then the next node is the black hole. Thus, the
agent has solved BHS. Otherwise, if the agent finds a node with a
single token this is the homebase of some agent (maybe itself). The
agent puts the second token on this node and continues the walk
without any token (i.e. it imitates the cautious-walk). If it again
encounters a node with a single token, then the next node is the black
hole and the algorithm terminates.  Otherwise, the agent keeps a
counter initialized to one and increments the counter whenever it
encounters a node containing two tokens. When the counter reaches a
value of six, the procedure terminates. At this point the agent is on
a node with two tokens (which it can use for the next procedure).

Unless an agent enters or locates the black hole, the procedure \emph{First-Ring} requires exactly $6nD$ time units for an agent that is alone in the ring, $3nD$ for two agents that start on the same ring, and $2nD$ if all the three agents start on the same ring.

\bigskip

\noindent \textbf{Procedure} \emph{Init-Next-Ring}:\\
\noindent An agent executes this procedure at the start of procedure
\emph{Next-Ring} in order to choose its new homebase for exploring the
next ring. The general idea is that the agent checks the node $u$ on
the South of its current location, move its two tokens to the East,
and then goes back to $u$. If there is another agent that has started
\emph{Next-Ring} on the West of $u$ (i.e., without this Procedure, the
homebases of the two agents would have formed a BHC), the agents can
detect it, and \emph{Init-Next-Ring} is designed in such a way that
the two agents meet. More precisely, when an agent executes
\emph{Init-Next-Ring} on horizontal ring $i$ without falling into the black hole, we ensure that
either (i) it meets another agent, or (ii) it locates the black hole,
or (iii) it detects that the black hole is on ring $i+2$, or (iv) the
token it leaves on its homebase does not form a BHC with a token on
ring $i+1$.  In case (iii), the agent executes
\emph{Black-Hole-in-Next-Ring}; in case (iv), it continues the
execution of \emph{Next-Ring}.

\bigskip

\noindent \textbf{Procedure} \emph{Next-Ring}:\\
\noindent The agent keeps one token on the homebase and with the other
token performs a special cautious-walk during which it traverses the
safe ring and at each node it puts a token, goes South to check the
node below, returns back and moves the token to the East. The agent
keeps a counter initialized to zero, which it increments whenever it
sees a node with a token on the safe ring. When the agent sees a token
on the safe ring, it does not go South, since this may be
dangerous. Instead, the agent goes West and South, and if it does not
see any token there, it puts a token and goes East. If the agent enters
the black hole, it has left a BHC. When the counter reaches a value of
six, the procedure terminates.

During the procedure, an agent keeps track of how many (1 or 2) tokens
it sees in the safe ring and the ring below. This information is
stored as a sequence (of length at most 24). At the end of the
procedure, using this sequence, an agent in the horizontal
ring $i$ can detect whether (i) the Black hole lies in the horizontal
ring $i+1$ or $i+2$, or, (ii) there are two other agents in the ring
$i$ and ring $i+1$ respectively, or, (iii) the ring $i+1$ does not
contain the black hole. In scenario (i), the agent executes procedure
\emph{Black-Hole-in-Next-Ring}; in scenario (ii), the agent meets with
the other agent in the same ring; in scenario (iii), the agent moves
to the next horizontal ring (i.e. ring $i+1$) to start procedure
\emph{Init-Next-Ring} again.

\bigskip

\noindent \textbf{Procedure} \emph{Black-Hole-in-Next-Ring}:\\ The
agent executes this procedure only when it is certain that the
black hole lies in the ring below its current position. The procedure is similar to \emph{Next-Ring};  the main difference being that the agent does not leave a homebase token. During the
procedure, either (i) the agent detects the location of the black hole
and marks all links to the black hole or (ii) the agent falls into the
black hole, forming a BHC at the black hole.






\subsection{Formal description of the algorithm}\label{sec-app-algo}

We now present in details the different procedures that are used in
Algorithm~\ref{algo:BHS-T-3-2}. As explained before, the agents use
procedure \FR to explore the horizontal ring where they start
and \NR to explore the others horizontal rings. 

\begin{algorithm}
\tcc{ Algorithm for BHS in Oriented Torus (using k=3 agents, 2 tokens)}
\SetKw{and}{and}
\SetKw{go}{Go}
\SetKwFunction{Wait}{Wait}
\SetKwFunction{NextRing}{NextRing}
\SetKwFunction{MarkAll}{Mark-All}
\SetKwFunction{CautiousWalk}{CautiousWalk}
\SetKwFunction{Exit}{Exit}
\SetKwFunction{FirstRing}{FirstRing}
\FirstRing\;
\NextRing{true}\;
 \caption{BHS-Torus-32}
  \label{algo:BHS-T-3-2}
\end{algorithm}

In the following, we sometimes denote nodes by their coordinates in
the ring. The $North$ (resp. $East$, $South$, $West$) neighbor of node
$(i,j)$ is the node denoted by $(i-1,j)$ (resp. $(i,j+1)$, $(i+1,j)$,
$(i,j-1)$).

\medskip
Recall that for both the procedures \emph{First-Ring} and
\emph{Next-Ring}, a big-step is the period between when an agent arrives
at a node $v$ with its token(s) to when it has left node $v$ with its
token(s).  Note that in both procedures, a big-step takes the same number
of time units that we denote by $D$. 

\medskip
\noindent\textbf{FirstRing.}
During this procedure the agent explores the horizontal ring that
contains its starting location.  The agent puts one token on its
homebase and uses the other token to perform cautious-walk in the
direction East, until it enters a node with some tokens. This node is
the homebase of some agent $a$ (maybe itself). If there are two tokens
on this node, then it means died on its first move going $East$.
Thus, the agent has solved BHS. Otherwise, the agent puts the second
token on this node and continues the walk using cautious-walk moves
but without moving tokens. Note that, in this case, agent $a$ has
already explored these nodes, and so it is safe for the agent to
continue going $East$ until it reaches a token. If it again encounters
a node with a single token, then it must be the second token of agent
$a$, because otherwise, $a$ would have put its second token on top of
this token. Thus, it implies that the next node is the black hole. 
Otherwise, the agent can only see nodes with two tokens on this ring. The agent keeps a counter
initialized to one and increments the counter whenever it encounters a
node containing two tokens. When the counter reaches a value of six,
the procedure terminates. At this point the agent is back on its
homebase with its two tokens (which it can use for the next
procedure).

When an agent locates the black hole, it marks all links incident to
it, and then it uses the Procedure \texttt{CleanFirstRing} in order to
remove all tokens that have been left on the homebases of the agents. 

Unless an agent dies or locates the black hole, the
procedure \emph{First-Ring} requires exactly $6n$ big-steps (i.e.,
$6nD$ time units) for an agent that is alone in the ring, $3n$
big-steps (i.e., $3nD$ time units) for two agents that start on the
same ring, and $2n$ big-steps (i.e., $2nD$ time units) if all the
three agents start on the same ring.

\begin{algorithm}
\tcc{Algorithm for the first horizontal ring. If two (or three) agents start in the horizontal ring where the black hole is located, the black hole is found.} 
\SetKw{and}{and}
\SetKw{go}{Go}
\SetKwFunction{Wait}{Wait}
\SetKwFunction{NextRing}{NextRing}
\SetKwFunction{MarkAll}{Mark-All}
\SetKwFunction{CautiousWalk}{CautiousWalk}
\SetKwFunction{Exit}{Exit}
\SetKwFunction{Clean}{CleanFirstRing}

Put 2 tokens \;
\Repeat{$n > 0$}
{
  reset $clock$\;
  \go $East$\;
  \go $West$\;
  pick up a token\;
  \go  $East$\;
  $n := $ the number of tokens you see\;
  \lIf{$n = 2$}{    
    \MarkAll{East}, \Clean and $\Exit$\tcc*{The black hole is found}}
  \lElse{Put a token\;}  
  \Wait until $clock$ reaches $magic\_number$ and reset $clock$\;
}
$count:=1$\;
\Repeat{$count=6$}
{
  reset $clock$\;
	\go $East$\;
	\lIf{you see 1 token}
	{
		\MarkAll{East}; \Clean and $\Exit$\tcc*{The black hole is found}
	}
	\lIf{you see 2 tokens}
	{
		$count:=count+1$\;
	}
   \Wait until $clock$ reaches $magic\_number$ and reset $clock$\;     
}
Pick up $2$ tokens \;
 \caption{FirstRing}
  \label{algo:FirstRing}
\end{algorithm}

\begin{procedure}
\tcc{Procedure to use if the Black-Hole is in the first ring in order
  to remove all tokens from this ring except the one on the West of
  the Black-Hole}   
\SetKw{go}{Go}
\SetKwFunction{Wait}{Wait}
\SetKwFunction{NextRing}{NextRing}
\SetKwFunction{MarkAll}{Mark-All}
\SetKwFunction{CautiousWalk}{CautiousWalk}
\SetKwFunction{Exit}{Exit}
\SetKwFunction{FirstRing}{FirstRing}
\Repeat{until you see exactly one token or you meet an agent}
{
  \go West\;
  \lIf{you see some tokens}{pick them up}\;
}
\caption{CleanFirstRing()}
\label{algo:Clean}
\end{procedure}

\medskip
\noindent\textbf{If two agents meet.} 
If two agents meet at a node, then it is quite easy to locate the
black hole using the team of two agents (without the help of any third
agent). This algorithm is described below (see procedure
{\BHStogether}). Notice first that it is always possible to break the
symmetry between the agents who meet at a node, because the agents
would arrive from different directions. In our
algorithm, if two agents meet, they are both close from their two
tokens and they first go to pick their tokens up. Once each agent has
its two tokens, we do the following.
After meeting, one of the agent becomes the leader and the other is
the follower. Together they perform a combined cautious walk
(Procedure \CWtogether) described as follows. The follower stays at
the current node while the leader goes to the next node and returns
immediately if the node is safe. Then both agents move together to the
next node. The starting node is marked with three tokens (recall the
agents have two tokens each, thus four tokens in total). The
cautious-walk is repeated until the agents come back to the starting
node\footnote{Note that there can be at most one node in the torus
that contains three tokens.}. The leader goes to the node directly
below to check if this node is the black hole. If not, the agents now
move to the next ring below along with the tokens, and repeat the
whole process. Only the leader agent may fall into the black hole and
in that case, the follower knows this fact within two time units, and
thus it has located the black hole.

\begin{procedure}
\SetKw{go}{Go}
\SetKwFunction{CWWAA}{Cautious-Walk-With-Another-Agent}
\tcc{Procedure for BHS in Oriented Torus using 2 colocated agents,} 
Pick up any token\; 
\go back to the home base of the first agent with both agents\;
Pick up any token\;
\go back to the home base of the second agent with both agents\;
Pick up any token\;
\Repeat{until black hole is found}
{
	Put 3 tokens\;
	\Repeat{three tokens found}
	{
		\CWWAA{East}\;		
	}
	Pick up 3 tokens\;
	\CWWAA{South}\;	
}

\caption{BHS-with-colocated-agents()}
  \label{proc:BHS-with-colocated-agents}
\end{procedure}

\begin{procedure}
\SetKwFunction{MarkAll}{Mark-All}
\SetKwFunction{Wait}{Wait}
\SetKwFunction{Exit}{Exit}
\SetKw{go}{Go}
\SetKw{and}{and}
first agent moves to direction and go back\;
second agent \Wait{2}\;
\If{second agent does not see first agent}
{\MarkAll{direction} \and \Exit\;}
{both agents move to direction}
\caption{Cautious-Walk-with-another-agent(direction)}
  \label{proc:Cautious-Walk-with-another-agent}
\end{procedure}

\begin{remark}
In our algorithm, as soon as two agents meet, they
execute \BHStogether. Note that, in this case, if the third agent sees
tokens belonging to the two agents executing \BHStogether, it sees a
tower of $3$ tokens. Since in our algorithm, as long as no agents have
met, there is at most two tokens on each node. Thus, the third agent
can detect that the two other agents have met, and in this case, it
stops the algorithm. We also assume that once two agents have met, if
they meet the third agent while executing \BHStogether, then the third
agent also stops executing the algorithm. 
\end{remark}

\begin{remark}
While executing the algorithm, if an agent visits a node which one of
its incident links is marked as dangerous, then the agents stops
executing the algorithm.  
\end{remark}

\begin{procedure}
\SetKw{go}{Go}
\SetKwFunction{BHScoloc}{BHS-with-colocated-agents}
\SetKwFunction{InitNextRing}{InitNextRing}
\SetKwFunction{NextRing}{NextRing}
\SetKwFunction{MarkAll}{Mark-All}
\SetKwFunction{Exit}{Exit}
\SetKwFunction{Wait}{Wait}
\SetKwFunction{BHNextRing}{BlackHoleInNextRing}
\SetKwFunction{Anal}{Analyze}
\SetKw{and}{and}
\SetKw{ore}{or}
\tcc{At any time during the execution, if you meet an agent you call procedure
\BHScoloc }
reset $clock$\;
\lIf{not $first\_time$}{\go South}
\lElse {\Wait{1}\tcc*{to ensure all agents start \InitNextRing at the
    same time.}} 
\InitNextRing\;
$count := 0$; $sequence := \epsilon$; $danger := false$\;
\Repeat
    {$count = 6$}
    {
      \Wait{12}\;
      \tcc{You wait to enable an agent executing \InitNextRing on the
        ring above to meet you if needed.}
      \eIf{danger}
          {\Wait{2}\;
            $danger := false$\;}
          {
            \go $South$\;
            $w:=$ the number of tokens you see\;
            \If{$w > 0$}
               {
                 $sequence:=sequence\concat b_w$\;
               }
               \go $North$\;
          }
      \lIf{$count<3$}{Pick up 1 token\;}
      \go $East$\;
      $n:=$ the number of tokens you see\;
      \lIf{$count <3$}{Put 1 token\;}
      \If{$n > 0$}
          {
            $count := count + 1$\;
            $sequence:=sequence\concat t_n$\;
            \eIf{$count \leq 3$}
               {\lIf{$n = 2$ \ore $w = 2$}
                      {\MarkAll{South} and $\Exit$\tcc*{The black
                          hole is found}}
                \lElseIf{$n=1$ and $w=1$}{$danger := true$\;}      
                \Else
                    {
                      Pick up 1 token\;
                      \go $West$, \go $South$\;
                      Put 1 token\;
                      \go $East$\tcc*{If you die, there is a token
                        North and West of the Black Hole}
                      \go $West$\;
                      Pick up 1 token\;
                      \go $North$, \go $East$\;
                      Put 1 token\;          
                    } 
               }
               {\lIf{$w \geq 1$}{$danger := true$\;}}
          }    
        \Wait until $clock$ reaches $magic\_number$ and reset $clock$\;
    }
Pick up all tokens\;
\Anal{sequence}\;
       
\caption{NextRing(first\_time)}
    \label{proc:NextRing}
\end{procedure}

\medskip
\noindent\textbf{NextRing.} Once an agent has finished exploring the
horizontal ring where it started executing the algorithm, it uses
Procedure \NR to explore the other rings.  An agent executing \NR on
ring $i$ knows that ring $i$ is safe, and it wants to explore the
nodes of ring $i+1$. To do so, the agent first executes \InitNR that
we describe below. If an agent $a$ continues executing \NR after it
has executed \InitNR, we know that if $a$ is on node $(i,j)$, then
there is no token on node $(i+1,j-1)$, and thus $a$ can safely leave a
token on $(i,j)$ without creating a BHC with the node
$(i+1,j-1)$. Then, the agent does a special cautious walk, i.e., it
repeats the following moves until it meets a token on ring $i$: It
leaves a token on its current node on ring $i$, goes $South$, goes
$North$, picks up its token and goes $East$. If the agent dies, then it
could have only died when it went $South$, and in this case, it has left a
token on the node on the $North$ of the black hole. In case the agent safely reaches the node on the $South$, if it sees some tokens, it remembers how many tokens it
sees.  

When agent $a$ has reached a node $v$ on ring $i$ where there is one
or two tokens, $a$ remembers how many tokens it sees.  Either $v$ is
the homebase of some agent $b$ (that may be the same as $a$), or the
token (or the two tokens) on $v$ indicates that the black hole is on
the $South$ or on the $East$ of this node. However, if $a$ is
executing \NR on ring $i$, it implies that ring $i$ is safe and thus
the black hole cannot be on ring $i$. Thus, either $v$ is the homebase
of an agent $b$, or the black hole is on the $South$ of $v$. 

If there are two tokens on $v$, then $v$ is the homebase of some agent
$b$ and the black hole is on the $South$ of $v$; in this case, $a$
locates the black hole.  If there is only one token on $v$, $a$ cannot
safely go $South$. However, we would like to check if the black hole
is on the $South$ of $v$. To do so, $a$ goes $West$ and then $South$
with its token; let $u$ be the node reached (note that $u$ is on
the $South-West$ of $v$). If there is no token on $u$, $a$ leaves its
token, and then goes $East$. If the black hole is on the $South$ of
$v$, $a$ dies leaving a black hole configuration.  If the black hole
is not on the $South$ of $v$, $a$ picks up its token and goes back to
$v$.

If there is one or two tokens on $u$, it is a black hole
configuration, and thus $a$ cannot safely go $East$. If there are two
tokens on $u$, then necessarily the black hole is on the $South$ of
$v$, and $a$ locates it. However, if there is one token on $u$, and
one token on $v$, it does not necessarily means that the black hole is
on the $South$ of $v$. Indeed, suppose that $v$ is the homebase of
$a$, that the black hole is on the $South$ of $u$, and that an agent
$c$ has started executing \NR on ring $i+1$ at the same time $a$
started executing \NR on ring $i$. After $v$ has executed \InitNR,
there was no token on $u$, but $c$ has died leaving its token on $u$
later.  Thus, if there is one token on $v$, and one token on $v$, $a$
continues to execute $\NR$.

If $a$ has neither died, nor located the black hole at $v$, it
continues to perform its special cautious walk until it sees some
tokens on ring $i$ twice. Each time it sees some tokens on ring $i+1$,
$v$ remembers how many tokens it sees. Each time $a$ sees some tokens
on ring $i$, it remembers how many tokens it sees and checks if the
black hole is on the $South$ as explained above. 

Note that if the black hole is on ring $i+1$, it implies that $v$ is
not the homebase of $a$, and that another agent died leaving a token
on $North$ of the black hole. Consequently, if $v$ dies, it enters the
black hole from the $East$. 

Once $a$ has seen three times some tokens on ring $i$, we can show
that if the black hole is on ring $i$, either $a$ died, or $a$ located
the black hole, or there is a BHC around the black hole. Thus, agent
$a$ leaves its second token on top of the token at its current node.
Then it performs a special cautious walk, avoiding entering nodes
marked by a BHC, until it sees tokens on ring $i$ three more
times. Again, while doing this, it remembers how many tokens it saw on
nodes on rings $i$ and $i+1$. However, during this final traversal, whenever $a$ reaches a node on ring
$i$ that contains one or two tokens, it does not check if the node on
the $South$ is the black hole.

\begin{remark}
In order to remember how many tokens it sees on the nodes of rings $i$
and $i+1$, the agent builds a sequence over the alphabet
$\{b_1,b_2,t_1,t_2\}$. Initially, its sequence is empty; each time the
agent sees one (resp. two) token on ring $i+1$, it adds a $b_1$
(resp. $b_2$) to its sequence; each time the agent sees one
(resp. two) token on ring $i$, it adds a $t_1$ (resp. $t_2$) to its
sequence.

We can show that when an agent sees some tokens on ring $i+1$, these
are either tokens left by dead agents, or homebase-tokens. This
implies that the sequence is of length at most $24$: an agent
with finite memory can remember such a sequence.

In the description of the algorithm, the $\concat$ symbol stands for
the standard concatenation of string and $\epsilon$ for the empty
string. 
\end{remark}

\begin{procedure}
\SetKw{go}{Go}
\SetKwFunction{BHScoloc}{BHS-with-colocated-agents}
\SetKwFunction{InitNextRing}{InitNextRing}
\SetKwFunction{NextRing}{NextRing}
\SetKwFunction{MarkAll}{Mark-All}
\SetKwFunction{Exit}{Exit}
\SetKwFunction{Wait}{Wait}
\SetKwFunction{BHNextRing}{BlackHoleInNextRing}
\SetKwFunction{Anal}{Analyze}
\SetKw{and}{and}
\SetKw{ore}{or}
\uIf{$sequence =b_1t_1b_1t_1b_1t_1b_2t_2b_2t_2b_2t_2$ $\ore$
  $sequence$ does not contain any $b$}
    {
      \tcc{Either you have seen no tokens on the ring below or you are
        a single agent that has seen tokens of another alive single
        agent} 
      \go $South$\;
      \NextRing($false$)\;
    }
\uElseIf{$sequence$ contains less than $3$ $t_2$}
       {
         \eIf{you have only one token}
             {\tcc{In this case, the two other agents were in the same
                 ring as you, they both died, and there is only one
                 node in the ring containing two tokens: the node that
                 is $North$ of the black hole} 
              \Repeat{you see two tokens}{\go $East$}
              \MarkAll $South$ and \Exit\;
             }
             {\tcc{In this case, there was another agent in the ring
                   with you, but you are the  only one that is still alive} 
              \BHNextRing\;
              }
       }
\uElseIf{$sequence$ contains two consecutive $t$}
       {\tcc{You know that there is another active agent  in the ring,
           and that your sequence is different from the sequence of the
           other agent.}
	 \eIf{$sequence$ start with $b$}
	        {\Wait until you meet an agent}
	        {
	          \Repeat{you meet an agent}{\go $East$}
	        }
       }
\Else
    {
         \tcc{You know that the next ring is safe and that an agent dies
           exploring the ring below it.}
         \go $South$\;
         \BHNextRing\;  
    }
\caption{Analyze(sequence)}
\end{procedure}

\medskip
\noindent\textbf{How to use the sequence constructed during \NR?}
At the end of Procedure \NR, the agent $a$
executing \NR calls Procedure \Analyze. This procedure enables an
agent to distinguish which of the following cases happen (see
Lemma~\ref{lem-next-ring}).

\begin{itemize}

\item $a$ does not see any tokens on ring $i+1$ and ring $i+1$ is
safe; in this case $a$ executes \NR on ring $i+1$.

\item there is no other agent on ring $i$ and there is an agent that
  has executed \NR without dying on ring $i+1$ when $a$ was
  executing \NR on ring $i$; in this case, $a$ executes \NR on ring
  $i+1$.

\item there were two other agents executing \NR on ring $i$,
  the black hole is on ring $i+1$, and $a$ is the only agent that is
  still alive; in this case, $a$ locates the black hole.

\item there was another agent executing \NR on ring $i$, the black
  hole is on ring $i+1$, and $a$ is the only agent that is still
  alive; in this case, $a$ executes \BlackHoleInNextRing on ring $i$. 

\item there is another agent on ring $i$, ring $i+1$ is safe and the
  tokens the agents see on ring $i+1$ enable the two agents to meet.

\item there is no other agent on ring $i$, black hole is in ring
  $i+2$, and there are the tokens of one or two dead agents on ring
  $i+1$; $a$ executes \BlackHoleInNextRing on ring $i+1$ with its two
  tokens. 

\end{itemize}

\medskip
\noindent\textbf{InitNextRing.}
The aim of Procedure \InitNR is to ensure that when an agent $a$
start executing \NR on ring $i$, the homebase of $a$ does not form a
BHC with a token on ring $i+1$. 

In Figure~\ref{fig-init-nr}, we have shown the different
possible trajectories for an agent executing \InitNR. The figure can be read as follows:
\begin{itemize}
\item the double-circled node is the place where the agent starts
  executing \InitNR. 
\item the black node, if any, represents the black hole.
\item the numbers in black on top of each node represent the time
  units where the agents arrive on the node. 
\item the numbers in red in the node represent the number of tokens
  belonging to other agents that the agent sees when it visits the
  node. If the agent visits the node twice and if it does not see the
  same number of tokens, we write $x/y$ that means that it sees
  $x$ tokens the first time and it sees $y$ tokens the second time. 
\item the intervals in green below each node represent the nodes where
  the agent left its own tokens; note that any agent executing
  \InitNR always moves its two tokens together. If an interval $[x - y]$ is
  written below a node, it means that the agent left its two tokens on
  this node between time units $x$ and $y$. When an agent left its
  tokens on a node $x$ time units after it started executing \InitNR
  and that they are not moved before the end of \InitNR, we write
  $[x -]$. 
\end{itemize}

\begin{figure}
\scalebox{0.90}{
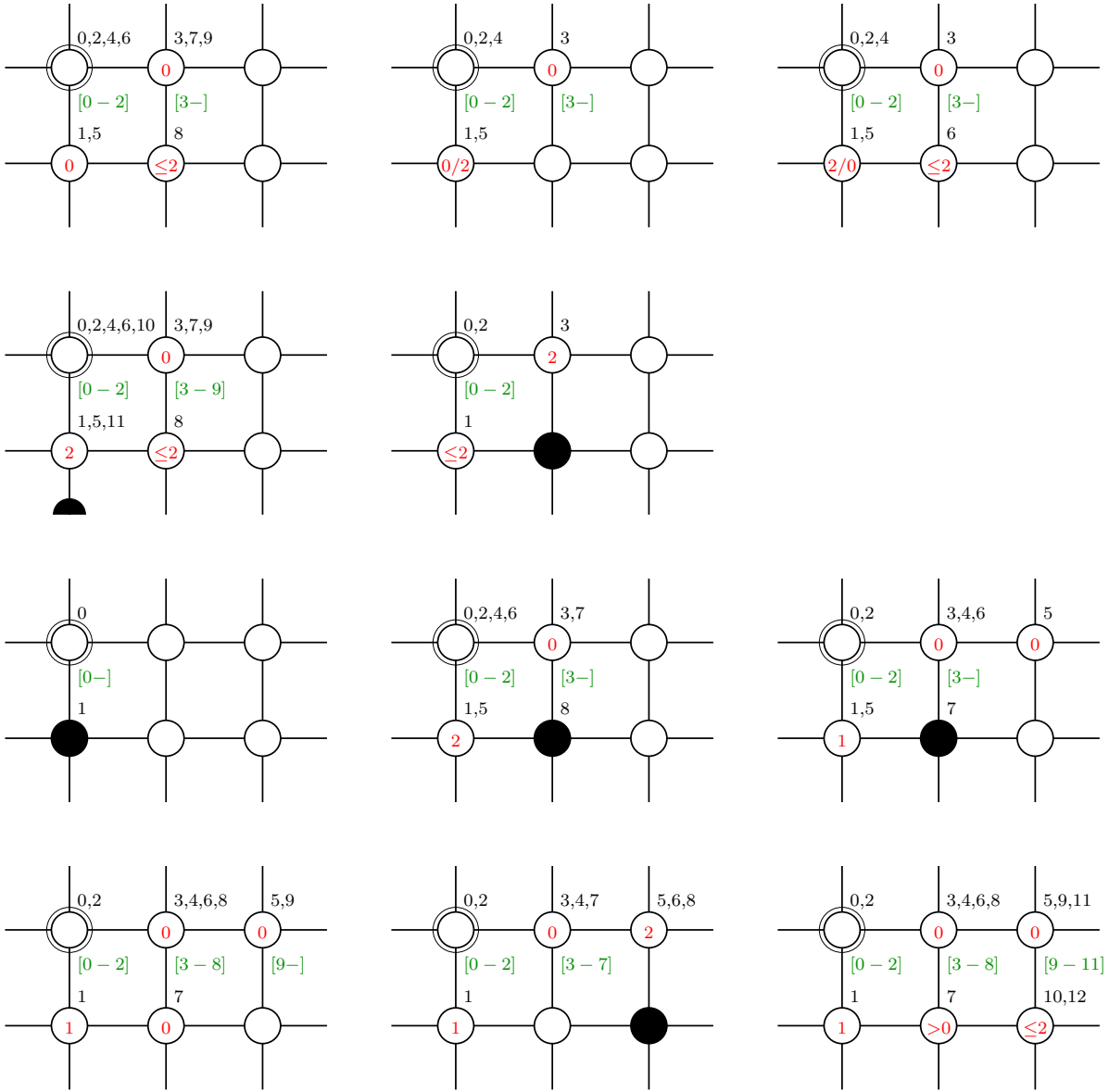  
}
\caption{The different trajectories an agent can follow while
  executing \InitNR.}\label{fig-init-nr}
\end{figure}

First, $a$ leaves its two tokens on the node $(i,j)$ where it
starts \InitNR, and then it goes $South$ on node $(i+1,j)$ and
remembers how many tokens it sees. Then the agent moves its two tokens
$East$ on node $(i,j+1)$. If the agent sees two tokens on node
$(i,j+1)$, then it implies that the black hole is on node $(i+1,j+1)$
and agent $a$ locates it. 

Otherwise, if the agent has seen zero or two tokens on node $(i+1,j)$,
it goes back to the node and checks how many tokens it sees. If $a$
has seen zero tokens both times, we know that no agent has started
executing \InitNR at the same moment as $a$ did. In this case, $a$
enters node $(i+1,j+1)$ from the $North$ (where it has left its two
tokens) to meet an agent $b$ that may be waiting there (if there is
another agent in ring $i$, and if $b$ is in the middle of the
execution of \NR); if it does not meet an agent, it goes back on node
$(i,j+1)$.

If the agent has seen twice two tokens on node $(i+1,j)$, it means
that the black hole is either on node $(i+1,j+1)$, or on node
$(i+2,j)$. In this case, agent $a$ enters node $(i+1,j+1)$ from the
$North$ (where it has left its two tokens); if $a$ dies, it leaves a
black hole configuration, and otherwise, the black hole is on node
$(i+2,j)$ and $a$ locates it. 

If the agent has seen first zero (resp. two), and then two (resp. zero)
tokens on node $(i+1,j)$, then it implies that another agent $b$ has
started executing \InitNR on node $(i+1,j-1)$ (resp $(i+1,j)$) at the
same moment as $a$ did. In this case, $a$ waits for $b$ on the two
tokens $b$ left on node $(i+1,j)$ (resp. $(i+1,j+1)$).

\begin{procedure}

\SetKwFunction{BHScoloc}{BHS-with-colocated-agents}
\SetKwFunction{NextRing}{NextRing}
\SetKwFunction{INR}{InitNextRing}
\SetKw{and}{and}
\SetKw{go}{Go}
\SetKwFunction{MarkAll}{Mark-All}
\SetKwFunction{Wait}{Wait}
\SetKwFunction{Exit}{Exit}
\SetKwFunction{BHNextRing}{BlackHoleInNextRing}
\SetKwFunction{OTB}{OneTokenBelow}

\tcc{Procedure used by the agent to proceed to the black hole search in the next ring.\\ At any time during the execution, if you meet an agent you call procedure
\BHScoloc }
Put 2 tokens\;
 \go $South$\;
$n_1 :=$ the number of tokens you see\;
{
	\go $North$\;
         Pick up 2 tokens\;
         \go $East$\;
	\lIf{you see 2 tokens}
            {\MarkAll{South} and \Exit\;}
           \lElse{Put 2 tokens\;}
}
\eIf{$n_1 = 1$}
{
  \OTB\;
}    
{       \go $West$, \go $South$\;
	$n_3 :=$ the number of tokens found\tcc*{same place as for
          $n_1$}
\Switch{the value of $(n_1,n_3)$}
{
	 \Case{(0,0)}
	 {
	 \go $North$, \go $East$\tcc*{back to your new homebase}
         \go $South$, \go $North$\tcc*{if you die, you leave 2
           tokens on top of the Black-Hole}  
	\Wait{3}\tcc*{to ensure all agents finish \INR at the same time}
	 }
	 \Case{(0,2)}
	 {
	 \tcc{You know that there is someone in the ring below you that did exactly 
 			the same thing as you did and it will come back to its homebase}
		\Wait until you meet an agent\;
	}
	 \Case{(2,0)}
	 {
	 \tcc{You know that there is someone in the ring below you that did exactly 
 			the same thing as you did and it will go back to its home base at East} 
                \go $East$\;
		\Wait until you meet an agent\;
	}
	\Case{(2,2)}
	{
		\tcc{An agent have died and have left two tokens next to the black hole}
		\go $North$, \go $East$, \go $South$\;
		\tcc{If you die, you leave the good configuration indicating the black hole and the third agent will find it}
                \go $North$\;
                Pick up 2 tokens\;
                \go $West$\;
		\go $South$\;
		\lIf{no agent is waiting} {\MarkAll{South} and \Exit}\;
	}
}
}
\caption{InitNextRing()}
  \label{proc:InitNextRing}
\end{procedure}

The last case to consider is when agent $a$ has seen one token on node
$(i+1,j)$; in this case, agent $a$ executes Procedure \OneTokenBelow.
Since the agents are synchronized, we can show that this token belongs
to a dead agent $b$: the black hole is either in ring $i$, or in ring
$i+1$. In this case, agent $a$ first enters node $(i+1,j+1)$ from the
$North$ (where it has left its two tokens); if $a$ dies, it leaves a
black hole configuration, and otherwise, it remembers how many tokens
it sees on node $(i+1,j+1)$. If it does not see any token on node
$(i+1,j+1)$, then $a$ moves its two tokens on node $(i,j+2)$: it knows
that its homebase-token will not form a BHC with a node on ring
$i+1$. If $a$ sees at least one token on node $(i+1,j+1)$, it moves
its two tokens on node $(i,j+2)$ and enters node $(i+1,j+2)$ from the
$North$. If $a$ died, it leaves a black hole configuration. Otherwise,
it knows that the tokens $a$ saw on nodes $(i+1,j)$ and $(i+1,j+1)$
are either homebase-tokens, or tokens indicating that the black hole
is on ring $i+2$. However, if the black hole is on ring $i+1$, then it
implies that both tokens $a$ saw were homebase tokens, and thus two
agents have executed \FR on ring $i+1$. But in this case, we know that
the black hole has already been found and all homebase tokens of ring
$i+1$ have been removed. Consequently, if $a$ sees one token on node
$(i+1,j)$, and at least one token on node $(i+1,j+1)$, then the
black hole is in ring $i+2$. In this case, agent $a$
executes \BlackHoleInNextRing on ring $i+1$.

\begin{procedure}

\SetKwFunction{BHScoloc}{BHS-with-colocated-agents}
\SetKwFunction{NextRing}{NextRing}
\SetKwFunction{INR}{InitNextRing}
\SetKwFunction{OTB}{OneTokenBelow}
\SetKw{and}{and}
\SetKw{go}{Go}
\SetKwFunction{MarkAll}{Mark-All}
\SetKwFunction{Wait}{Wait}
\SetKwFunction{Exit}{Exit}
\SetKwFunction{BHNextRing}{BlackHoleInNextRing}

\tcc{Procedure used by the agent that sees one token during InitNextRing. When the agent starts executing this procedure, it is one node east and one node north from the node containing one token. \\ 
At any time during the execution, if you meet an agent you call procedure
\BHScoloc }

\Wait{1}\;
\go $East$\;
\eIf{you see 2 tokens}
    {\Wait{1}\tcc*{If an agent is doing the same thing, you wait for it.} 
      \go $West$; Pick up 2 tokens; \go $East$\;
          {\MarkAll{South} and \Exit\;}
    }
    {
      \go $West$, \go $South$\;
      \tcc{if you died you leave the good configuration indicating the
        black hole and the third agent will find it.}
      $n_2 := $ the number of tokens you see\;
      \go $North$\;
      Move 2 tokens to the $East$\;
      \eIf{$n_2 > 0$}
         {
           \go $South$\;
           \tcc{If you died you leave the good configuration
             indicating the black hole and the third  
	        agent will find it.}
              \go $North$\;
              \tcc{There are two nodes with tokens in the ring below you and for
                both marked positions, the link to the $East$ is
                safe. Thus, the ring below you is safe and the
                black hole is on the ring below this ring.}
              Pick up $2$ tokens\;
              \go $South$\;
              \BHNextRing\;
         }
         {\Wait{3}\tcc*{to ensure all agents finish \OTB and \INR at
             the same time.}} 
    }
\caption{OneTokenBelow()}
  \label{proc:one-token-below}
\end{procedure}

\medskip
\noindent\textbf{BlackHoleInNextRing. } 
When it executes \InitNR, or \NR, an agent may locate the ring
containing the black hole without locating the black hole itself. In
this case, it executes the procedure \BlackHoleInNextRing. When an
agent $a$ executes \BlackHoleInNextRing on ring $i$, agent $a$ has its
two tokens, it knows that ring $i$ is safe and that the black hole is
on ring $i+1$. First, agent $a$ reaches a node on ring $i$ that does
not contain any token; this ensures  that agent is not on the $North$
of the black hole. Then, agent $a$ traverses the ring $i$ until it
finds a token on node $(i,j)$. If there is a token on node $(i+1,j-1)$
($a$ can check this safely by moving $West$ and then $South$), i.e.,
agent $a$ discovers a black hole configuration, then we can show that
the black hole is on node $(i,j)$, and agent $a$ locates it. If there
is no token on node $(i+1,j-1)$, $a$ leaves its two tokens on node
$(i+1,j-1)$ and enters node $(i+1,j)$ going $East$. If $a$ dies, it
leaves a black hole configuration. Otherwise, agent $a$ repeat this
procedure until it dies or locates the black hole. Since we know that
there is a token on the node located on the $North$ of the black hole,
we know that either $a$ locates the black hole, or $a$ dies entering
the black hole from the $East$ leaving a black hole configuration.

\begin{procedure}
\SetKw{go}{Go}
\SetKwFunction{BHScoloc}{BHS-with-colocated-agents}
\SetKwFunction{InitNextRing}{InitNextRing}
\SetKwFunction{NextRing}{NextRing}
\SetKwFunction{MarkAll}{Mark-All}
\SetKwFunction{Exit}{Exit}
\SetKwFunction{Wait}{Wait}
\SetKw{and}{and}
\SetKw{ore}{or}

\tcc{An agent executes this procedure when it knows that the ring it
  is moving in is safe, and that the Black Hole is in the ring below.\\
  At any time during the execution, if you meet an agent you call procedure
\BHScoloc}

\While{you see some tokens}
      {\go $East$\;}

\Repeat
    {you find the Black Hole}
    {
      \Repeat
          {you see some tokens}
          {\go $East$\;}
      \tcc{The Black Hole may be South; one needs to check this.}
      \go $West$\;
      \go $South$\;
      \eIf{you see some tokens}
      {
             \MarkAll{East}\;\tcc{You located the Black Hole since no agent had come back to pick the token.}
      }
      {
        Put 2 tokens\;
        \go $East$
        \tcc*{If you die, there are some tokens North and West of the
          Black Hole.}
        \go $West$\;
        Pick up 2 tokens\;
        \go $North$, \go $East$\;
      }
    }
\caption{BlackHoleInNextRing()}
  \label{proc:BH-NextRing}
\end{procedure}


\subsection{Proof of Correctness for Algorithm BHS-Torus-32}\label{sec-app-proof}

\begin{remark}
During the execution of our algorithm, once two agents have met, they
execute \texttt{BHS-with-co-located-agents}, and they eventually find
the black hole. We are mainly interested in showing that even if 
agents do not meet, our algorithm is correct.

In a lot of the following lemmas, we implicitly assume that each agent
that is still alive has not met any other agent.
\end{remark}

\begin{lemma}\label{lem-cautious-walk}
Assume that when an agent starts executing \InitNR on ring $i$, each
other alive agent is either waiting with its two tokens on its
homebases, or it starts executing \InitNR, or it starts executing
\BlackHoleInNextRing. Moreover, assume that ring $i$ is safe, and that
on every node of ring $i$ that contains tokens, there is an agent on
its two tokens.

When an agent $a$ is the only agent on a node $v$, if it
leaves $v$ to visit a node $w$ that may be the black hole, the
following holds:
\begin{itemize}
\item $a$ always leaves $v$ going $East$, or $South$,
\item $a$ always leaves ones or two tokens on $v$,
\item all the tokens on $v$ when $a$ left belong to $a$,
\item if $w$ is not the black hole, the next move of $a$ is to go back to $v$.
\end{itemize}
\end{lemma}

\begin{proof}
The only times that an agent that is alone goes to a node that maybe
the black hole are during the execution of \FR, \NR,
\BlackHoleInNextRing or \InitNR. When an agent executes \FR, it can
only enter the black hole from the $East$, and this can only happen
before it has seen any token. When an agent executes the main loop of
\NR on ring $i$, agent $a$ can die from the $North$ only before it
meets a token on ring $i$. Thus, the only moves where $a$ can enter
the black hole is when it executes Line 4 of Procedure \FR, Lines 10
(when $count = 0$) and 29 of Procedure \NR, Line 2, 17, 25 of
Procedure \InitNR, Lines 8 and 13 of \OneTokenBelow, and Line 13 of
Procedure \BlackHoleInNextRing.  One can check that in all these
cases, all the properties hold.
\end{proof}


\begin{lemma}\label{lem-first-ring}
If two or three agents start in the horizontal ring containing the
black hole, then the black hole is found during the exploration of the
first ring. 

Moreover, if two agents start in the horizontal ring containing the
black hole, then the agent that locates the black hole removes all
tokens left on the homebases of the agents before the third agent
visits any node of this ring.

If one or more agents starts on a ring that does not contain the black
hole, they finish executing \FR on the vertex where they start
without marking any link.
\end{lemma}

\begin{proof}
Suppose first that two or three agents starts in the horizontal ring
containing the black hole.  Consider two agents $a$ and $b$ such that
when moving $East$ on the horizontal ring where the agents starts, $a$
is the closer to the black hole, and $b$ is the second closer.

Note that $b$ cannot die while executing the loop between lines $2$
and $11$ of Procedure \FR, since it sees some tokens when it
arrives on the homebase of $a$.

If $a$ dies on its first move to the $East$, then $a$ left two tokens
on its homebase. Otherwise, $a$ comes back and picks up a token before
continuing its cautious walk, and thus $b$ sees only one token when it
arrives on the homebase of $a$. 
Consequently, the first time $b$ sees some token, the black hole is
located to the $East$ of its current position if and only if $b$ sees
exactly two tokens. 

Suppose the black hole is not located immediately on the $East$ of the
homebase of $a$. Agent $a$ dies while executing the loop between lines
$2$ and $11$ of Procedure \FR and leaves a token on the node on
the $West$ of the black hole. Thus, $b$ first visits the homebase of
$a$ where it sees one token, and then visit a node with one token: it
will mark the black hole links.

If there are three agents on the ring, let $c$ be the third
agent. Then it is easy to see that one of the following happens:
\begin{itemize}
\item either the third agent meets $b$ while $b$ is performing its
  cleaning phase, or once $b$ has finished it,
\item or $c$ reaches the homebase of $b$ before $b$ has terminated its
  cleaning phase.
\end{itemize}

In the second case, $c$ first visit a node with one token (the
homebase of $b$), puts a token on top of it and then continue going
$East$, and thus $b$ picks up all the tokens that have been left on
the homebases of $a$, $b$ and $c$. If $c$ arrives on the homebase of
$a$ before $b$ has picked up the tokens, then $c$ sees two tokens on
the node and continue going $East$ to reach the node where $b$ marked
the $East$ link as leading to the black hole. If $c$ arrives on the
homebase of $a$ once $b$ has picked up the tokens, then after it has
visited the homebase of $b$, it continues going $East$ and reach the
node where $b$ marked the $East$ link as leading to the black hole.

Suppose now that there is one or more agents starting in a ring that
does not contain the black hole. 

If there is only one agent in the ring, each time the agent sees some
token, it is back on its homebase: the first time, there is only one
token on its homebase and it adds a token. All the other times, it
sees two tokens. Thus, the agent performs $6$ turns of the ring during
the execution of \FR.

If there are two agents $a$ and $b$ in the ring, the first time $a$
sees some tokens, it is on the homebase of $b$ and sees only one
token. The next time it sees some token, it is back on its homebase,
but $b$ has arrived there before and left its second token on the
node. The next time $a$ sees some tokens, it is successively on the
homebases of $b$, $a$, $b$, and $a$. Consequently, $a$ is back on its
homebase and has performed $3$ turns of the ring during the execution
of \FR.

When there are three agents $a$, $b$, $c$ starting in the same ring,
let assume that when we traverse the ring going $East$ starting from
the homebase of $a$, we reach the homebase of $b$ before the homebase
of $c$.  For the same reasons as before, no agent has marked any link
as leading to the black hole. And when agent $a$ successively sees
some tokens, it is successively on the homebases of $b$, $c$, $a$,
$b$, $c$ and $a$. Consequently $a$ is back on its homebases when
$count = 6$ and it has performed $2$ turns of the ring during the
execution of \FR.

When two agents start on the ring containing the black hole, it takes
less than $2n$ big-steps to the surviving agent to locate the black hole
and to remove tokens left on homebases. An agent that is alone in its
ring needs $6n$ big-steps to finish executing \FR if the black hole is not
in its ring. Thus, if two agents start on the ring $i$ containing the
black hole, the only token the third agent can see on ring $i$ is the
token located on the West of the black hole. Note that when the third
agent reaches this node, the link going $East$ has been marked as
leading to the black hole and thus the agent stops executing the
algorithm.
\end{proof}

\begin{lemma}
Consider two alive agents $a$ and $b$ and assume the black hole has
not been found yet. When $a$ starts the execution of \InitNR on some
ring $i$, either $b$ is also starting the execution \InitNR on some
ring $j$, or $b$ is starting the execution of \BlackHoleInNextRing, or
$b$ is in the middle of the execution of \NR, and is waiting on its
homebase with its two tokens.
\end{lemma}

\begin{proof}
We prove the lemma by induction on the numbers of rings explored by
agent $a$, and we distinguish different cases. 
 
First suppose that all agents have started in the same ring. If an
agent dies during \FR, the black hole is found by the other agents. It
takes exactly $2n$ big-steps (i.e., $2nD$ time units) to each agent to
execute \FR or \NR on each ring. Consequently, while no agent is dead,
all agents always start the execution of \NR simultaneously. Note that
if an agent dies while executing \NR, it enters the black hole from
the $North$, and either another agent find the black hole when it sees
the tokens of the dead agent, or it dies entering the black hole from
the $East$.

Now assume that all agents have started in different rings. As long as
no agent is dead, it takes exactly $6n$ big-steps (i.e., $6nD$ time units)
to each agent to execute \FR or \NR on each ring. Thus, if agent $a$
and $b$ are neither dead, nor executing \BlackHoleInNextRing, they
starts executing \NR simultaneously.

The last case to consider is when two agents $a$ and $b$ start in the
same ring, while the third agent $c$ is in another ring. We also show
that the three agents will never execute \NR in the same ring.

While agents $a$ and $b$ are not in the same ring as $c$, it takes
exactly $3n$ big-steps (i.e., $3nD$ time units) to agents $a$ and $b$ to
execute \FR or \NR on each ring, while it takes $6n$ big-steps to agent
$c$. Consequently, as long as all agents are not in the same ring,
each time agent $c$ starts \NR, agents $a$ and $b$ start \NR at the
same time, unless they are dead, or executing
\BlackHoleInNextRing. 

Suppose now that the three agents are not in the same ring and that
$a$ starts the execution of \NR. If $b$ is not dead, or executing
\BlackHoleInNextRing, $b$ starts executing \NR at the same
time. Consider now agent $c$, and assume that it is still alive and
that it is not executing \BlackHoleInNextRing.  Let $q$ be the number
of times $a$ has executed \FR or \NR so far, i.e., $a$ has performed
$3nq$ big-steps since it has started executing the algorithm. If $q$ is
even, then agent $c$ has executed \FR or \NR $q/2$ times, and starts
executing \NR at the same time as $a$. Otherwise, $c$ has executed \FR
or \NR $(q-1)/2$ times and has performed $3n$ big-steps of \NR. In this
case, since $c$ is alone in its ring, it means $c$ is back in its
homebase with its two tokens, and waiting for $D$ time units. 

We now show that the three agents cannot be in the same ring when they
start executing \NR. Since agent $a$ and $b$ are twice as fast as $c$,
and since the agents start in two different rings, there will be a
big-step where agent $c$ is executing \NR on ring $i+1$ while agents $a$
and $b$ are executing \NR on ring $i$. Assume that agent $c$ does not
meet any other agent during \InitNR.  If $c$ does not die while
executing \NR on ring $i+1$, then assume without loss of generality
that agent $a$ sees the token $c$ left on its base before agent
$b$. If agent $c$ starts \NR on ring $i+1$ at the same time as agents
$a$ and $b$ start \NR on ring $i$, agent $c$ carries one of its token
to perform a special cautious walk while it has left its other token
on its homebase. If agent $c$ is in the middle of executing \NR on
ring $i+1$ when agents $a$ and $b$ start \NR on ring $i$, both tokens
of agent $c$ are on its homebase.  In both cases, since the tokens of
$a$ and $b$ always stay on ring $i$, agents $a$ and $b$ can see tokens
on exactly one node of ring $i+1$. 

Then the sequence $a$ builds while executing \NR starts with
$b_1t_1t_1$ (or $b_2t_1t_1$), while the sequence of $b$ starts with
$t_1b_1t_1$ (or $t_1b_2t_1$). Due to the design of Procedure \Analyze,
this implies that agents $a$ and $b$ will not execute \NR on ring
$i+1$.
\end{proof}

\begin{lemma}\label{lem-nextring-emptyring}
When an agent starts \NR on ring $i$, it knows that ring $i$ is
safe, and that on every node of ring $i$ that contains tokens, there
is an agent on its two tokens.
\end{lemma}

\begin{proof}
Consider an agent $a$ that starts executing \NR on ring $i$ at time
$t$. First assume that this is the first time agent $a$ executes
\NR. From Lemma~\ref{lem-first-ring}, we know that the black hole is
not in ring $i$ (otherwise, $a$ is either dead or has located the
black hole). In this case, we know that all agents in ring $i$ are
back on their homebases with their two tokens. Moreover, since the
agents are synchronized, the only case to consider is when there are
two agents $b$ and $c$ that started in ring $i-1$. While $a$ executed
\FR, agents $b$ and $c$ have executed \FR and a first iteration of
\NR. However, during the execution of \NR, the only tokens $b$ and $c$
see on ring $i$ are the two tokens of $a$ on its homebase and thus the
sequence computed by $b$ and $c$ are $b_2t_1t_1b_2t_1t_2b_2t_2t_2$ and
$t_1b_2t_1t_1b_2t_2t_2b_2t_2$. In this case $b$ and $c$ have different
sequences and they meet without leaving a unique token on ring $i$.

Assume now that agent $a$ has already executed \NR on ring
$i-1$. Suppose that there exists a unique token on a node. Since all
agents are synchronized, all alive agents are with their two tokens at
time $t$ (either before starting \NR, or in the middle of the
execution of \NR). Thus, agent $a$ has executed \NR on ring $i-1$, and
$a$ has seen this token on ring $i$ while executing \NR on ring $i-1$.
Thus, its sequence is different from $t_1^3t_2^3$. Moreover, if the
sequence of $a$ at the end of the execution of \NR on ring $i-1$ is
$(b_1t_1)^3(b_2t_2)^3$, it means that there are at least $6$ other
tokens on the ring $i$ (three towers of $2$) after it sees the token
at position $(i,j+1)$. Since we have only three agents, this is
impossible.
\end{proof}


\begin{lemma}\label{lem-2-east}
When executing \InitNR on ring $i$, starting on node $(i,j)$, if an
agent sees some tokens on node $(i,j+1)$ without meeting an agent,
then the black hole is located on node $(i+1,j+1)$, and the agent sees
$2$ tokens on this node.
\end{lemma}

\begin{proof}
Consider an agent $a$ executing  \InitNR in ring $i$, starting
at position $j$ and that sees some tokens at position $j+1$ on ring
$i$ (line 7 of Procedure \InitNR). 

We first prove that agent $a$ sees two tokens on position $(i,j+1)$.
Suppose that there is only one token on position $(i,j+1)$. Since the
agent are synchronized, this token belongs to a dead agent $b$.  Since
an agent is moving its two tokens together during \InitNR, agent $b$
has died before $a$ started \InitNR. But, from
Lemma~\ref{lem-nextring-emptyring}, we know that this is impossible. 


Thus, agent $a$ sees two tokens at position $(i,j+1)$.  If these
tokens belong to a dead agent, then the black hole is either East or
South of this node.  Since agent $a$ is performing \InitNR on
ring $i$, it knows ring $i$ is safe, and thus, if there is a black
hole, it has to be the South node. 

Suppose these tokens belong to an agent $b$ that is still alive. Since
the agents are synchronized, either the agent is waiting on its
homebase, or the agent has started the execution of \InitNR at
the same moment as $a$ did. In the first case, the two agents meet. In
the second case, it would mean that agent $b$ and $a$ started the
execution of \InitNR on the same node; which is impossible.
\end{proof}

Consider an agent that starts executing \InitNR on node $(i,j)$. While
executing \InitNR, if the agent sees two (resp. zero) tokens the first
time it goes to node $(i+1,j)$ and sees zero (resp. two) tokens the
second time it visits node $(i+1,j)$, we say that the agent sees two
tokens appearing (resp. disappearing). 

\begin{lemma}\label{lem-moving-tokens}
During \InitNR, if an agent sees tokens appearing or
disappearing, then two agents meet, or the black hole is located.
\end{lemma}

\begin{proof}
Suppose that an agent $a$ executing \InitNR on ring $i$ at
position $j$, sees $0$ tokens (resp. $2$ tokens) the first time it goes
on ring $i+1$ at position $j$ and $2$ tokens (resp. $0$ tokens) the
second time it goes at position $(i+1,j)$. Since the agents are
synchronized, all alive agents are either waiting at their homebases
with their two tokens, or executing \InitNR. Since the tokens
have appeared (resp. disappeared), we know that there was an alive
agent $b$ that executes the lines $1$ to $6$ of Procedure
\InitNR.  Consequently, we know that ring $i+1$ is safe and thus
agent $a$ can move East if the tokens have disappeared.

Thus, we know that $5$ (resp. $6$) time units after the beginning of
the execution of \InitNR, agent $a$ is waiting on the two tokens
located at the East of the homebase of agent $b$.  These two tokens
can be the two tokens of a dead agent, or the tokens of $b$ that $b$
moved during \InitNR. In the first case, agent $b$ locates the
black hole (see Lemma~\ref{lem-2-east}). 

Otherwise, if when going $South$, agent $b$ sees either twice $0$
tokens, or twice $2$ tokens, then $b$ is back on its two tokens $7$
time units after the beginning of \InitNR. If when going
$South$, agent $b$ sees $1$ token, then $b$ cannot die before the $8$
time units after the beginning of \InitNR. Moreover, either
agent $b$ meets the third agent at time $6$ on the node located at the
$East$ of the node where $a$ is waiting, or $b$ is back on its two
tokens at time $6$ or $7$. Thus the two agents meet. 

Suppose now that agent $b$ sees first $2$ (resp. $0$) tokens the first
time it goes $South$ and $0$ (resp. $2$) the second time. In this
case, it means that there is an agent $c$ executing \InitNR.
Consequently, we know that ring $i+2$ is safe. 
Thus, $a$, $b$ and $c$ are the three agents executing the algorithm,
and they are located on lines $i$, $i+1$, $i+2$. If, when executing
\InitNR, agent $c$ sees tokens appearing, or disappearing, then it
means that these tokens are the tokens of agent $a$ (since there is
only $3$ agents executing the algorithm). Thus the torus has three
horizontal lines and all of them are safe, which is
impossible. Consequently, agent $c$ does not see tokens appearing, or
disappearing, while it executes \InitNR, and thus agents $b$ and $c$
meet.
\end{proof}

\begin{lemma}
When executing \InitNR on ring $i$ at position $j$, if an agent
sees twice two tokens on ring $i+1$ at position $j$, then one of the
following holds:
\begin{itemize}
\item either, the black hole is on ring $i+1$ at position $j+1$, and
  the agent dies leaving a black hole configuration, 
\item or the black hole is on ring $i+2$ at position $j$ and the agent
  locates it,
\item or two agents meet.
\end{itemize}
\end{lemma}

\begin{proof}
Suppose that an agent $a$ executing \InitNR on ring $i$ at
position $j$, sees $2$ tokens the two times it goes on ring $i+1$ at
position $j$. 

There are two cases to consider: either the two tokens that $a$ sees
the first time have not been moved, or an agent $b$ picked up these
two tokens, while an agent $c$ moved its two tokens to this node.

In the second case, it means that agents $b$ and $c$ are alive when
agent $a$ starts \InitNR. Since the agents are synchronized,
both agents are also executing \InitNR and thus ring $i+1$ is
safe.  Since the three agents are located on ring $i$ and $i+1$,
agents $b$ and $c$ do not see any token on ring $i+2$ and since they
both moved their tokens to the $East$ of their starting positions,
none of them died going $South$. Thus, they are back on their tokens
$7$ time units after the beginning of Procedure \InitNR. 
Since ring $i+1$ is safe, agent $a$ is back on ring $i+1$ at position
$j$, $11$ time units after it started executing Procedure
\InitNR: $a$ meets another agent. 

Suppose now that the two tokens seen by agent $a$ the second time are
the same as the two tokens it sees the first time. Since the agents
are synchronized, all agents that are still alive are either waiting
on their two tokens, or are executing \InitNR and have picked up their
tokens $2$ time units after they started \InitNR. If $a$ does not meet
another agent the first time it goes South, it implies that the agent
that put these two tokens is dead. From Lemma~\ref{lem-cautious-walk},
we know that the black hole is located either on ring $i+2$ at
position $j$, or on ring $i+1$ at position $j+1$.

Assume that agent $a$ does not meet any agent while it executes
\InitNR. If the black hole is located at position $j+1$ on ring
$i+1$, agent $a$ dies while going to this node, after it left its two
tokens on ring $i$ at position $j+1$. Thus, its two tokens and the two
tokens on ring $i+1$ at position $j$ form a black hole configuration. 
Otherwise, agent $a$ does not die while executing \InitNR and
locates the black hole that is on ring $i+2$ at position $j$. 
\end{proof}

\begin{lemma}\label{lem-1-south}
When executing \InitNR in ring $i$, if an agent sees one token
the first time it goes $South$, it knows another agent is dead, and
one of the following holds.
\begin{itemize}
\item either it meets another agent, 
\item or it locates the black hole,
\item or it dies leaving a black hole configuration, 
\item or it knows ring $i+1$ is safe and the black hole is in ring
  $i+2$, and executes \BlackHoleInNextRing on ring $i+1$,
\item or it continues executing \NR without leaving a black hole
  configuration with a token on ring $i+1$.
\end{itemize}
\end{lemma}

\begin{proof}
Consider an agent $a$ executing \InitNR on ring $i$, starting at
position $j$.  Since the agents are synchronized, when agent $a$
executes \InitNR, all alive agents are either waiting on their
homebases, or executing \InitNR and moving their two tokens
together.  Thus, if agent $a$ sees a token on the node, then it is the
token of a dead agent.  This token may be the homebase token of the
dead agent, or a token indicating that the black hole is located at
the $South$ or at the $East$ of this token. In any case, the black
hole is in ring $i+1$, or in ring $i+2$.

\begin{claim}\label{claim-above}
If an agent $b$ starts \InitNR on ring $i-1$, at position $j$,
$j+1$ or $j+2$, $a$ and $b$ meet.
\end{claim}

If an agent $b$ starts at position $j$, or $j+1$, the claim follows
from Lemma~\ref{lem-moving-tokens}. Suppose now that an agent starts
at position $j+2$ on ring $i$. From
Lemma~\ref{lem-nextring-emptyring}, we know that $b$ cannot see a
unique token on ring $i$ at position $j+2$. Since ring $i$ is safe,
agent $b$ does not see two tokens on ring $i-1$ at position $j+3$, and
thus, agent $b$ goes back at position $(i,j+2)$ $5$ time units after
it started \InitNR. Since agent $a$ is on position $(i,j+2)$ at this
moment, the two agents meet.

Assume now that no agent starts \InitNR on ring $i-1$, at
position $j$, $j+1$, or $j+2$. 

\begin{claim}
If an agent $b$ starts on ring $i$ at position $j+1$, either the black
hole is at position $(i+1,j+1)$ and $a$ locates it, or $a$ and $b$
meet.
\end{claim}

If the black hole is at position $(i+1,j+1)$, then $b$ dies leaving
its two tokens on node $(i,j+1)$. From Lemma~\ref{lem-2-east}, $a$
locates the black hole. 

Suppose now that the black hole is not at position $(i+1,j+1)$.  First
note that since both $a$ and $b$ have their tokens with them when they
start \InitNR, and since there is one token on node $(i+1,j)$,
from Lemma~\ref{lem-2-east}, agent $b$ cannot see any token at
position $(i,j+2)$.  Moreover, $b$ can either see $0$ or $1$ token at
position $(i+1,j+1)$. If $b$ does not see a unique token on node
$(i+1,j+1)$, then $b$ is at position $(i,j+1)$ $4$ time units after it
started \InitNR; Since at this moment, $a$ is also on this node,
the two agents meet. Suppose now that $b$ sees one token at position
$(i+1,j+1)$. Since $b$ cannot see any token on node $(i,j+3)$, $b$ is
at position $(i,j+2)$ $6$ time units after it started
\InitNR. Since $a$ sees $2$ tokens on node $(i,j+2)$ $5$ time
units after it started \InitNR, it waits there one time unit,
and thus the two agents meet on this node.

Assume now that agent $a$ does not meet any agent while executing
\InitNR. 
\begin{claim}
If agent $a$ sees some tokens at position $(i,j+2)$, then $a$ sees $2$
tokens, the black hole is located at node $(i+1,j+2)$, and $a$ locates
it.
\end{claim}

For the same reasons as in the proof of Lemma~\ref{lem-2-east}, since
agent $a$ is executing \InitNR, $a$ can either see $0$ or $2$
tokens on node $(i,j+2)$. 

Suppose the two tokens located at node $(i,j+2)$ belong to an agent
$b$ that is still alive. Since the agent are synchronized, either $b$
is waiting on its two tokens, or agent $b$ has started executing
\InitNR at the same moment $a$ did. In the first case, $a$ meets
$b$. In the second case, agent $b$ has picked up its two tokens and
moved to the $East$ $3$ time units after it started
\InitNR. Consequently, it implies that agent $b$ started on node
$(i,j+1)$, but in this case, we know from the previous claim that $a$
and $b$ meet. 

Thus, the two tokens $a$ sees on node $(i,j+2)$ belong to a dead
agent. From Lemma~\ref{lem-cautious-walk}, the black hole is either
located at the $South$ or at the $East$ of this node. Since $a$ is
executing \NR on ring $i$, we know that ring $i$ is safe, and
thus the black hole is located at node $(i+1,j+2)$. Since $a$ does not
meet any other agent, it locates the black hole.

Assume now that agent $a$ does not meet any other agent, and does not see
any token on nodes $(i,j+1)$ and $(i,j+2)$. 

\begin{claim}
If the black hole is located on node $(i+1,j+1)$ or $(i+1,j+2)$, $a$
dies leaving a black hole configuration.
\end{claim}

If the black hole is on node $(i+1,j+1)$, agent $a$ dies $7$ time
units after it started \InitNR leaving its two tokens on node
$(i,j+1)$. With the unique token on node $(i+1,j)$, this leaves a
black hole configuration. 

If the black hole is on node $(i+1,j+2)$, then the token $a$ sees on
node $(i+1,j)$ is a homebase token of some dead agent $b$. From
Lemma~\ref{lem-cautious-walk}, $b$ left a token on node $(i+1,j+1)$
that $a$ sees  at time $7$. Thus, agent $a$ dies at time $10$ after it
left its two tokens on node $(i,j+2)$. Consequently, with the token on
node $(i+1,j+1)$, $a$ dies leaving a black hole configuration. 

Assume now that while it executes \InitNR, agent $a$ does not meet
any other agent, and does not see any token on ring $i$. Furthermore,
assume that the black hole is neither on node $(i+1,j+1)$, or
$(i+1,j+2)$. 

\begin{claim}
If agent $a$ sees one or two tokens on node $(i+1,j+1)$, then the
black hole is in ring $i+2$, and ring $i+1$ is safe.
\end{claim}

We already know that the black hole is either on ring $i+1$, or in
ring $i+2$. Since nodes located at position $(i+1,j+1)$ and
$(i+1,j+2)$ are safe, we know that each of these tokens is either a
homebase token, or a token indicating that the black hole is South. 

Suppose that the black hole is on ring $i+1$. Then, the tokens left on
both nodes are homebase tokens, and consequently either two agents
have started in the ring and at least one of them found the black
hole, or two agents died.  If two agents have started executing the
algorithm on ring $i+1$, then by Lemma~\ref{lem-first-ring}, the black
hole has been found, and while executing \NR on ring $i-1$ (or \FR on
ring $i$), agent $a$ visited a node such that the link to the $South$
was marked as leading to the black hole; i.e., $a$ is not executing the
algorithm any more.  Otherwise, at least one agent has executed \NR on
node $i$ before leaving its homebase on ring $i+1$, and thus it died
leaving its homebase token on ring $i$, and not on ring
$i+1$. Consequently, only one of the two tokens on positions $(i+1,j)$
and $(i+1,j+1)$ can be a homebase token, and consequently, the black
hole is on ring $i+2$, and thus ring $i+1$ is safe. In this case, $a$
executes \BlackHoleInNextRing on ring $i+1$. 


\begin{claim}
If agent $a$ continues the execution of \NR, then its homebase
token is not part of a black hole configuration with a token on ring
$i+1$ when it terminates \InitNR.
\end{claim}

If agent $a$ continues the execution of \NR, then it means that
$a$ does not meet any agent, does not see tokens on nodes $(i,j+1)$,
$(i,j+2)$, or $(i+1,j+1)$, and that it leaves a homebase token on node
$(i,j+2)$. Suppose that its token is part of a homebase
configuration.

Suppose that an agent $b$ leaves a token on node $(i+1,j+1)$ after
agent $a$ visited it. Then it means that agent $b$ is still alive when
agent $a$ starts \InitNR. Since the third agent is dead, $b$ is
either executing \BlackHoleInNextRing on node $i+1$, or \InitNR
on ring $i+1$. In the first case, $b$ does not leave any token on ring
$i+1$. In the second case, the agent is executing \InitNR on
ring $i+1$ and it means agent $b$ has visited node $(i+1,j+1)$ ans has
seen a unique token on the node; which is impossible from
Lemma~\ref{lem-2-east}. 
\end{proof}

\begin{lemma}\label{lem-nobhconf}
When executing \InitNR in ring $i$, if an agent does not see any
token and does not meet any agent, then it continues executing
\NR without leaving a black hole configuration with a token on
ring $i+1$.
\end{lemma}

\begin{proof}
Consider an agent $a$ that does not see any token and does not meet
any agent while executing \InitNR. If agent $a$ starts the execution
on node $(i,j)$, then it leaves its homebase token on node $(i,j+1)$,
and there was no token on ring $(i+1,j)$ when agent $a$ visited it. If
a token appears on node $(i+1,j)$, then there is an agent $b$
executing \InitNR on ring $i+1$ that has started on node $j-1$ or
$j-2$. In the first case, agent $a$ has seen the two tokens on agent
$b$ the first time it went $South$. In the second case, agent $b$
moves its two tokens to node $(i+1,j)$ only if it has seen a unique
token on node $(i+2,j-1)$. But in this case, from
Case~\ref{claim-above} of Lemma~\ref{lem-1-south}, both agents have
met.
\end{proof}

\smallskip
In the following lemma, we show that when an agent executes \NR, if it
has not meet any other agent during \InitNR, then there cannot be an
agent that is located just below it. This implies that while
executing \InitNR, an agent cannot see the token another agent uses
for its special cautious walk.

\begin{lemma}\label{lem-no-agent-below}
Once an agent has finished executing \InitNR and continues the
execution of \NR (i.e., it is not dead, it has not meet any other
agent, it is not executing \BlackHoleInNextRing), it knows that there
is no alive agent located immediately below it that executes \NR.
\end{lemma}

\begin{proof}
Consider an agent $a$ that started executing \InitNR on node $(i,j)$
at time $t$.  Once $a$ has finished executing \InitNR, it continues to
execute \NR only if it has seen no token on node $(i+1,j)$, or if it
sees one token on node $(i+1,j)$ and no token on node $(i+1,j+1)$. In
the first case, $a$ is on node $(i,j+1)$. Since $a$ did not see any
token on node $(i+1,j)$ and did not meet any agent, we know from
Lemma~\ref{lem-moving-tokens} that no agent has started executing
\InitNR on node $(i+1,j-1)$, or $(i+1,j)$ at time $t$. Thus, any agent
that has started executing \InitNR at time $t$ cannot continue
executing \NR and be on node $(i+1,j+1)$. Since the agents are
synchronized, we know that only an agent that was waiting on its
homebase with its two tokens can be on node $(i+1,j+1)$; but in this
case, agent $a$ meets it at time $t+8$. 

Suppose now that agent $a$ has seen exactly one token on node
$(i+1,j)$. In this case, $a$ is on node $(i,j+2)$ when it has finished
executing \InitNR, and we know that the token on node $(i+1,j)$
belongs to an agent $b$ that dies before $a$ started executing
\InitNR. Assume there is an agent $c$ that has finished executing
\InitNR on node $(i+1,j+2)$ and that neither $c$ has met any other
agent, nor $c$ has located the black hole, nor $c$ has started
executing \BlackHoleInNextRing. Note that this implies that $c$ has
started executing \InitNR on ring $i+1$.  If this is the first time
$a$ performs \InitNR, then $b$ died while performing \FR on ring
$i+1$, and $c$ also executed \FR on ring $i+1$. From
Lemma~\ref{lem-first-ring}, this implies that $c$ locates the black
hole while executing \FR. Consequently, $a$ and $c$ have already
respectively executed \NR on ring $i-1$ and on ring $i$. But in this
case, $c$ saw the token left by $b$ on ring $i+1$ while executing \NR
on ring $i$. Thus the sequence computed by $c$ contained at least one
$b_1$. Since $c$ has executed \NR on ring $i$, its sequence should
have been $(b_1t_1)^3(b_2t_2)^3$. But this means that $c$ saw at least
$6$ other tokens on the ring $i$ (three towers of $2$) after it sees
the token at position $(i+1,j)$. Since we have only three agents, this
is impossible.
\end{proof}

\begin{lemma}\label{lem-next-ring}
When an agent $a$ continue executing \NR on ring $i$ after it has
finished \InitNR, either $a$ dies entering the black hole, or $a$
locates the black hole, or $a$ does not see any token on ring $i+1$,
or $a$ has seen some token on ring $i+1$ (i.e., $sequence$ contains
$b_1$ or $b_2$), and then $a$ is in one of the following case and it
can detect in which case it is if from the sequence it constructed.
\begin{itemize}

\item there were two other agents executing \NR on ring $i$,
  the black hole is on ring $i+1$, and $a$ is the only agent that is
  still alive; $a$ locates the black hole.

\item there was another agent executing \NR on ring $i$, the black
  hole is on ring $i+1$, and $a$ is the only agent that is still
  alive; $a$ executes \BlackHoleInNextRing on ring $i$.

\item there is another agent on ring $i$, ring $i+1$ is safe and the
  tokens the agents see on ring $i+1$ enable the two agents to meet.


\item there is no other agent on ring $i$ and there is an agent that
  has executed \NR without dying on ring $i+1$ when $a$ has executed
  \NR on ring $i$; $a$ executes \NR on ring $i+1$.

\item there is no other agent on ring $i$, black hole is in ring
  $i+2$, and there are the tokens of one or two dead agents on ring
  $i+1$; $a$ executes \BlackHoleInNextRing on ring $i+1$ with its two
  tokens. 




\end{itemize}
\end{lemma}

\begin{proof}

Consider an agent $a$ that starts executing \NR on ring $i$ at time
$t$; note that from Lemma~\ref{lem-nextring-emptyring}, ring $i$ is
safe. If $a$ does not see any token on ring $i+1$, and if ring $i+1$
is safe, $a$ executes \NR on ring $i+1$ once it has terminated
executing \NR on ring $i$.

\begin{case}
The three agents start executing \NR on  ring $i$ at time $t$ and the
black hole is on ring $i+1$. 
\end{case}

Let $v$ be the node on the $North$ of the black hole, and $u$ the node
on the $East$ of the black hole. Without loss of generality, assume
that $c$ visits $v$ before $b$, and that $b$ visits $v$ before
$a$. Then, $c$ dies leaving one or two tokens on $v$, and $b$ dies
leaving one token on $u$. When $a$ arrives at $v$, if $c$ left two
tokens on $v$, $b$ has located the black hole. Otherwise, $v$
continues the execution of \NR without entering the black hole (it has
set its variable $danger$ to $true$). During the execution of \NR on
ring $i$, $a$ visits successively the homebase of $b$, the homebase of
$c$, $u$, $v$, the homebase of $a$, the homebase of $b$, the homebase
of $c$. It puts its second token on $v$ when it arrives in $v$ the
first time, leaving a black hole configuration. Thus the sequence $a$
has constructed is $t_1, t_1, t_1, b_1, t_1, t_1, t_1$. In this case,
$a$ goes to $v$ (this is the only node with two tokens) and locates
the black hole.

\begin{case}\label{case-2agents}
Ring $i+1$ is safe, two agents $a$ and $b$ start executing \NR on ring
$i$ at time $t$. 
\end{case}

Note that $a$ and $b$ cannot die while executing \NR on ring $i$.

First, assume that ring $i+2$ is also safe. Then the tokens $a$ and $b$
see on ring $i+1$ belong to agent $c$ that is alive all along the
execution of \NR on ring $i$ by agent $a$ and $b$. Since rings $i+1$
and $i+2$ are safe, agent $c$ is either executing \FR or \NR.  Suppose
first that $a$ and $b$ are executing \NR for the first time, i.e.,
they have just finished executing \FR and thus they have performed
exactly $3n$ big-steps. During these first $3n$ big-steps $c$ has visited its
homebase three times and it has put a second token on its homebase
when $a$ and $b$ start executing \NR. Moreover, it takes $3n$ big-steps to
$a$ and $b$ to execute \NR. During these $3n$ more big-steps, $c$ visits
its homebase $3$ times and put a second token on its homebase at the
end of these $3n$ big-steps. Consequently, during the whole execution of
\NR by $a$ and $b$, $c$ does not move its tokens. Suppose that $c$ is
in the middle of the execution of \NR on ring $i+1$ when $a$ and $b$
start the execution of \NR on ring $i$. Then $c$ has just put two
tokens on its homebase, and during the next $3n$ big-steps its tokens will
not move. If $c$ has started executing \NR on ring $i+1$ when $a$ and
$b$ start the execution of \NR on ring $i+1$, then during the $3n$
big-steps it takes $a$ and $b$ to perform \NR on ring $i$, $c$ visits its
homebase exactly three times and does not put a second token on its
homebase before the moment $a$ and $b$ have finished executing \NR on
ring $i$.

Without loss of generality, assume that $a$ visits the homebase of $c$
before $b$.  Consequently, when executing \NR, $a$ visits successively
the homebases of $c$, $b$, $a$, $c$, $b$, $a$, $c$, $b$, $a$, while
$b$ visits successively the homebases of $a$, $c$, $b$, $a$, $c$, $b$,
$a$, $c$, $b$. Each time $a$ or $b$ visits the homebase of $c$, they
either always see one token, or always see two tokens.

The first three homebases on ring $i$ that $a$ or $b$ sees contains
one token, and both $a$ and $b$ add a token on the third homebase on
ring $i$ they see. Thus, the last third homebases $a$ and $b$ visit on
ring $i$ contain two tokens. Consequently, the sequence of $a$ is
$tb_1b_1tb_1b_2tb_2b_2$, while the sequence of $b$ is
$b_1tb_1b_1tb_2b_2tb_2$ where $t$ is either $t_1$ or $t_2$. In this
case, $b$ waits on its homebase while $a$ moves along ring $i$ to meet
$b$ on its homebase. 

Suppose now that the black hole is in ring $i+2$. Since the agents $a$
and $b$ some tokens on ring $i+1$, it implies that $c$ has executed
\NR on ring $i+1$. We know that $c$ has started executing \NR on ring
$i+1$ at time $t$ or before. In any case, $c$ is either dead before
time $t$, or it dies during the first $n$ big-steps of \NR, i.e., before
it comes back to its homebase. Once $c$ is dead, there are two tokens
left on ring $i+1$ (they may be both on the homebase of $c$ if the
black hole is the node below). Without loss of generality, assume that
$a$ visits the homebase of $c$ before $b$. Let $v$ be the node where
is the second token of $c$, i.e., the node on top of the black
hole. 

If $a$ visits $v$ before it visits the homebase of $b$, then $c$ is
dead before $a$ and $b$ visit $c$ and the sequence of $a$ is $b_1^2t_1
^2b_1^2t_1t_2b_1^2t_2^2$ (or $b_2t_1 ^2b_2t_1t_2b_2t_2^2$ if $v$ is
the homebase of $c$) while the sequence of $b$ is $t_1
b_1^2t_1^2b_1^2t_2^2b_1^2t_2$ (or $t_1b_2t_1^2b_2t_2^2b_2t_2$ if $v$
is the homebase of $c$). In this case, $b$ waits on its homebase,
while $a$ moves on ring $i$ to meet $b$ on its homebase. 

If $a$ visits $v$ after it visits the homebase of $b$, then the
sequence of $a$ is $(b_1t_1)^3(b_1t_2)^3$.  Note that if $c$ has
started \NR on ring $i$ at time $t$, the first time $b$ visits $v$,
$c$ is not yet dead. Thus, the sequence of $b$ is either
$(b_1t_1)^3(b_1t_2)^3$, or $t_1(b_1t_1)^2(b_1t_2)^3$.  In this case,
both agents execute \BlackHoleInNextRing.

\begin{case}\label{case-2-plus1}
The black hole is in ring $i+1$, two agents $a$ and $b$
start executing \NR on ring $i$ at time $t$. 
\end{case}

Let $v$ be the node on the $North$ of the black hole and assume
without loss of generality that $b$ visits $v$ before $a$.

First suppose that no agent has explored ring $i+1$ yet. If $b$ has
started \NR on $v$, $b$ dies after its first move of \NR. In that
case, the first time $a$ sees some tokens while executing \NR is on
$v$ and it sees two tokens: $a$ locates the black hole. Otherwise, $a$
first sees the homebase token of $b$, and then the token $b$ left on
$v$. In this case, $a$ dies entering the black hole from the West.

Let $c$ be the third agent and assume that $c$ has already executed
\FR on ring $i+1$, or \NR on ring $i$. If $c$ has executed \NR on ring
$i$, it died leaving its tokens on ring $i$; but in this case, $a$ and
$b$ have executed \NR on ring $i-1$ and they have seen the tokens $c$
left. Thus, it implies that $c$ died while executing \FR. 

If $c$ died on its first move, then it left its two tokens on its
homebase, and thus $b$ dies entering the black hole from $v$, but it
leaves a black hole configuration, and $a$ will locate the black hole
while executing \NR. Otherwise, let $u$ be the node on the West of the
black hole. Then $b$ dies entering the black hole from $v$ and $a$ sees
one token on $u$ and one token on $v$ and thus it does not enter the
black hole. If we consider the tokens $a$ sees on ring $i$ while
executing \NR, $a$ visits successively the homebase of $b$, $v$, the
homebase of $a$, the homebase of $b$, $v$ and the homebase of $a$. It
adds a second token on its homebase the first time it visits it. Thus,
the sequence $a$ computed contained less than three $t_2$, and in this
case, $a$ will execute \BlackHoleInNextRing.

\begin{case}
Ring $i+1$ is safe, $a$ is alone in its ring and ring $i+2$ is safe.
\end{case}

Note that there cannot be two agents executing \NR on ring $i+1$ while
$a$ is executing \NR on ring $i$. Indeed, suppose $a$ has executed \NR
$q$ times, i.e., $a$ has performed $6(q+1)n$ big-steps, then if two agents
have started in the same ring, they have executed \FR once and \NR
$(12q+1)$ times, and thus they cannot be in the ring below $a$. 

Since both rings $i+1$ and $i+2$ are safe, the tokens $a$ sees on ring
$i$ are homebase tokens. If there is an agent $b$ that has executed
\NR on ring $i+1$ while $a$ is executing \NR on ring $i$, then between
two times $a$ goes back to its homebase, it sees the base of
$b$. Since the agents are synchronized, the third time $a$ is on its
homebase, $b$ is also on its homebase and both of them put a second
token on their homebases. Moreover, from previous lemmas, we know that
the two homebases of $a$ and $b$ do not form a black hole
configuration. Thus the sequence computed by $a$ is
$(b_1t_1)^3(b_2t_2)^3$. In this case, $a$ executes \NR on ring $i+1$.

\begin{case}
The black hole is in ring $i+1$ and $a$ is alone
in its ring.   
\end{case}

If $a$ is the first agent to explore ring $i+1$, $a$ dies entering the
black hole from the North.

If an agent $b$ has explored ring $i+1$ before $a$, then $b$ explored
the ring using \FR, because otherwise, $b$ would have left tokens on
ring $i$, and $a$ would have seen them while executing \NR on ring
$i-1$. In any case, $a$ dies entering the black hole from the North.

If two agents $b$ and $c$ have already explored ring $i+1$, they have
explored it using \FR and the black hole has been found.  

\begin{case}
Ring $i+1$ is safe, $a$ is alone in its ring, and the black hole is in
ring $i+2$. 
\end{case}

As before, we know it is not possible that two agents start executing
\NR on ring $i+1$ at time $t$. 

Since $a$ sees some tokens on ring $i+1$, we know that at least one
agent has started executing \NR on ring $i+1$ at time $t$ or before. 

First suppose that there is exactly one agent $b$ that has started
executing \NR on ring $i+1$ before time $t$, i.e, $b$ died while $a$
was executing \NR on ring $i-1$ or before. From the previous case, we
know that $b$ died entering the black hole from the $North$, leaving
one or two tokens on the node $v$ on $North$ of the black hole. If $b$
died leaving its two tokens on $v$, then during the execution of \NR,
$a$ visits $v$ and its homebase six times, and the sequence computed
by $v$ is $(b_2t_1)^3(b_2t_2)^3$. If $b$ died leaving only of of its
tokens on $v$, $a$ visits $v$, the homebase of $b$ and its homebase
six times, and the sequence it computes is
$(b_1b_1t_1)^3(b_1b_1t_2)^3$. In any case, $a$ executes
\BlackHoleInNextRing on ring $i+1$ once it has finished executing \NR
on ring $i$.

Suppose now that no agent executed \NR on ring $i+1$ before time $t$,
and that exactly one agent $b$ starts executing \NR on ring $i+1$ at
time $t$. First suppose that $b$ dies leaving two tokens on a vertex
$v$; then it implies that the black hole is on the node on $South$ of
$v$. Note that in this case, $b$ dies while executing \InitNR, and we
know that if $a$ starts \InitNR on node $(i,j)$ and $b$ starts \InitNR
on node $(i+1,j-1)$, or $(i+1,j)$, $a$ and $b$ meet. Moreover, since
$b$ is dead when $a$ finished \InitNR, we know from previous lemmas
that the two tokens of $b$ and the homebase token of $a$ do not form a
black hole configuration.  In this case, when executing \NR, $a$ never
sees a black hole configuration, $a$ visits six times $v$ and its
homebase, and the sequence it computes is $(b_2t_1)^3(b_2t_2)^3$. In
this case, $a$ executes \BlackHoleInNextRing on ring $i+1$ once it has
finished executing \NR on ring $i$.

Suppose now that $b$ does not die while it executes \InitNR, i.e.,
when $b$ dies, it has left one token on its homebase, and one token on
the node $v$ that is $North$ of the black hole. Note that the first
time $a$ visits $v$, $b$ may still be alive (if $a$ visits $v$ before
it visits the homebase of $b$), but the second time $a$ reaches $v$,
$b$ is dead and $a$ sees two tokens on $v$. During the execution of
\NR, $a$ visits successively six times $v$, the homebase of $b$, and
its homebase, or the homebase of $b$, $v$, and its homebase. In the
first case, the sequence computed by $a$ is
$b_1t_1(b_1b_1t_1)^2(b_1b_1t_2)^3$ while in the second case, the
sequence of $a$ is $(b_1b_1t_1)^3(b_1b_1t_2)^3$.  Moreover, from
previous lemmas, we know that the homebase tokens of $a$ and $b$ do
not form a black hole configuration. However it is possible that the
homebase token of $a$ and the token left on $v$ form a black hole
configuration. But in this case, there is one token on $a$, one token
on $v$ and thus, $a$ does not visit the node $u$ below its homebase
and continues executing \NR. Note that there cannot be a token on node
$u$, because, it can only be the homebase token of $b$, but this is
impossible from Lemma~\ref{lem-no-agent-below}. Thus, $a$ executes
completely \NR on ring $i$, and then it executes \BlackHoleInNextRing
on ring $i+1$. 

Suppose now that two agents $b$ and $c$ have already \NR on ring
$i+1$, starting at time $t$ or before. Since $a$ sees some tokens on
ring $i+1$, at least one of these two agents has executed \NR on ring
$i+1$.  Again, we distinguish different cases: either $b$ and $c$
started executing \NR on ring $i+1$ at the same time $t' \leq t$, or
$b$ has executed \NR on ring $i$ at time $t' < t$ while $c$ has
started executing \NR on ring $i+1$ at time $t'' \leq t$.

Suppose that $b$ and $c$ started executing \NR on ring $i+1$ at the
same time $t' \leq t$. For the same reasons as before, we know that
$t' < t$. From Case~\ref{case-2-plus1}, we know that either the black
hole has been found, or both agents are dead. If the black hole has
been found, the first time $a$ visits the node on $North$ of the black
hole, it sees that the link going $South$ has been marked, and it
stops the algorithm. If both agents are dead, we know that there is
one token on ring $i+2$ that is on the node on the West of the black
hole, and there are three tokens on ring $i+1$: the homebase of $b$,
the homebase of $c$, and the node on the $North$ of the black
hole. From Lemma~\ref{lem-no-agent-below}, we know $a$ cannot see a
black hole configuration while executing \NR on ring $i$. The sequence
computed by $a$ is then $(b_1^3t_1)^3(b_1^3t_2)^3$ and once $a$ has
finished executing \NR, it executes \BlackHoleInNextRing on ring
$i+1$. 

Suppose now that $b$ has executed \NR on ring $i$ at time $t' < t$
while $c$ has started executing \NR on ring $i+1$ at time $t'' \leq
t$. In this case $c$ dies while executing \NR on ring $i$, and we know
that $b$ start executing \BlackHoleInNextRing on ring $i+1$ at time
$t^* \leq t$. In this case, $c$ has left either one or two tokens on
the node $v$ on the North of the black hole, and $b$ does not leave
any token on ring $i+1$, and dies leaving its two tokens on the node
$u$ on the West of the black hole, thus leaving a black hole
configuration. In this case, the only tokens $a$ sees on ring $i+1$
while executing \NR on ring $i$ are the homebase token of $c$ and the
token left by $c$ on $v$. Since $c$ is already dead when $a$ starts
executing \InitNR on ring $i$, we know from previous lemmas that $a$
cannot see any black hole configuration when executing \NR on ring
$i$. Thus the sequence computed by $a$ is
$(b_1b_1t_1)^3(b_1b_1t_2)^3$, and $a$ executes \BlackHoleInNextRing on
ring $i+1$ once it has finished executing \NR on ring $i$.
\end{proof}


An agent executing \BlackHoleInNextRing first moves to find a place
where there is no token. In the following lemma, we prove that in any
ring there is always such a place. 

\begin{lemma}
If the black hole has not been found yet, when an agent $a$ starts
executing \BlackHoleInNextRing on ring $i$, then there exists a node
on ring $i$ where there is no token.

\end{lemma}

\begin{proof}

Recall that an agent $a$ is carrying its two tokens when it starts
\BlackHoleInNextRing on ring $i$ and it knows that ring $i$ is safe.


Note that on ring $i$, there can be at most $3$ nodes containing
tokens: two homebases, and one indicating the black hole is
South. Thus if the size of ring $i$ is at least $4$ we are done. 

Suppose now that the ring is of size $3$. Since $a$ is executing
\BlackHoleInNextRing, we know that either $a$ has executed \NR on ring
$i$, or on ring $i-1$. 

Suppose first that $a$ has started executing \NR on ring $i$ at time
$t$. Then we know from Lemma~\ref{lem-next-ring} that either the three
agents have executed \NR on ring $i$ at time $t$, or that another
agent $b$ have executed \NR on ring $i$ at time $t$ and that the third
agent $c$ was dead before time $t$.  In the first case, since the ring
is of size $3$, one agent died on its first move when executing
\InitNR, and in this case, the black hole has been found from
Lemma~\ref{lem-2-east}.  In the second case, we know from the proof of
Case~\ref{case-2agents} of Lemma~\ref{lem-next-ring} that $c$ died
while executing \FR, and thus its homebase token cannot be on ring
$i$.

Consequently, $a$ has executed \NR on ring $i-1$ at time $t$. It
implies that $a$ saw tokens of one or two dead agents during \InitNR,
or during the execution of \NR. Let $b$ and $c$ be the two dead agents
and assume that $c$ died before $b$. Since both $b$ and $c$ left their
tokens on ring $i$, both $b$ and $c$ died while executing \NR on ring
$i$. First suppose that $b$ and $c$ did not execute \NR on ring $i$
simultaneously. Then, when $c$ died, $b$ has not finished executing
\NR on ring $i-1$, and from Lemma~\ref{lem-next-ring}, $b$ saw the
tokens left by $c$, and $b$ has not executed \NR on ring $i$, but
\BlackHoleInNextRing on ring $i$. This implies that $b$ did not leave
its homebase token on ring $i$. Suppose now that $b$ and $c$ start
simultaneously the execution of \NR on ring $i$. If $c$ died on its
first move of \InitNR, then from Lemma~\ref{lem-2-east}, $b$ has
located the black hole. Otherwise, since the ring is of size $3$, $c$
died leaving its two tokens on its new homebase, and then $b$ locates
the black hole when it arrives in $c$. 
\end{proof}

\begin{lemma}
When an agent $a$ executes \BlackHoleInNextRing on ring $i$, at least
one agent is dead and the black hole is on ring $i+1$, and either $a$
locates the black hole, or $a$ died entering the black hole leaving a
black hole configuration.
\end{lemma}

\begin{proof}
We know that if $a$ executes \BlackHoleInNextRing on ring $i$, it has
executed \NR on ring $i$, or on ring $i-1$, and we know from
Lemmas~\ref{lem-1-south} and \ref{lem-next-ring} that ring $i$ is
safe, that the black hole is in ring $i+1$ and that $a$ saw tokens of
dead agents on ring $i$. Thus, at least one agent died while executing
\NR on ring $i$, and thus it left a token on top of the black hole.

Each time $a$ sees a token on node $(i,j)$, it goes to node
$(i+1,j-1)$. If $a$ does not see some tokens, it leaves its two tokens on
the node an enter node $(i+1,j)$ from the $East$.  Thus, if $a$ dies,
it leaves a black hole configuration. If $a$ sees some tokens on node
$(i+1,j-1)$, $a$ marks all links leading to the node $(i+1,j)$ as
leading to the black hole. 

Since we know that there is a token on the node on the North of the
black hole, we are sure that either $a$ dies, or $a$ marks links as
leading to the black hole.

Suppose $a$ is wrong when marking links going to node $(i+1,j)$ as
leading to the black hole. Since we know the black hole is in ring
$i+1$, this implies that the tokens on node $(i,j)$ and $(i+1,j-1)$
are respectively the homebase tokens of two dead agents $b$ and $c$.
This implies that $c$ died while executing \FR (otherwise its homebase
token would be on ring $i+1$), and that $b$ died while executing
\NR. Thus the token on node $(i+1,j-1)$ was already there when $b$
executed \InitNR on ring $i$. But from Lemmas~\ref{lem-1-south} and
\ref{lem-nobhconf}, we know it is impossible. Thus, if $a$ marks
links leading to node $(i+1,j)$, the black hole is indeed in node $(i+1,j)$. 
\end{proof}

\begin{lemma}
Suppose that two agents $b$ and $c$ are dead, and that the black hole
is on node $(i+1,j)$. Then, there are one or two tokens on both nodes
$(i,j)$ and $(i+1,j-1)$. Moreover, the agent $a$ that is still alive
never enters the black hole and eventually locates it.
\end{lemma}

\begin{proof}
Without loss of generality, assume that $c$ died before $b$.  We know
that agents $b$ and $c$ died in one of the following way:
\begin{itemize}
\item $c$ died executing \FR on ring $i+1$, and $b$ died
  executing \NR on ring $i$,
\item $b$ and $c$ died executing \NR on ring $i$,
\item $c$ died executing \NR on ring $i$ and $b$ died
  executing \BlackHoleInNextRing on ring $i$. 
\end{itemize}

In the first case, the first agent let one or two tokens on node
$(i+1,j-1)$ and the other agent let one or two tokens on node $(i,j)$.
In the second case, we know from Lemma~\ref{lem-next-ring} that it
implies that the two agents were executing \NR on ring $i$
simultaneously. Thus, $c$ died while performing its special cautious
walk and left a token on node $(i,j)$, and $b$ died when during the
execution of \NR on ring $i$, it saw the token left by $c$ and it
enters node $(i+1,j)$ after it left a token on node $(i+1,j-1)$.  In
the last case, from the previous lemma, we know that $b$ died leaving
a black hole configuration.

If $a$ is not executing \NR on ring $i$ with $b$ when $b$ dies, then
we know that $a$ executes eventually \NR on ring $i-1$.  When $a$
executes \NR on ring $i-1$, it sees some tokens on ring $i$, and from
Lemmas~\ref{lem-1-south} and \ref{lem-next-ring}, we know that $a$
executes \BlackHoleInNextRing on ring $i$. From the previous lemma, we
know that $a$ locates the black hole.

Suppose now that $a$ is executing \NR on ring $i$ with $b$ when $b$
dies. Once $b$ is dead, there are some tokens on nodes $(i,j)$ and
$(i+1,j-1)$. Thus, either $a$ locates the black hole, or it continues
to execute \NR avoiding the node $(i+1,j)$ (its variable $danger$ is
 $true$). We know from Lemma~\ref{lem-next-ring} that once $a$ has
finished executing \NR on ring $i$, $a$ executes \BlackHoleInNextRing
on ring $i$. From the previous lemma, $a$ locates the black hole. 
\end{proof}

\begin{theorem}\label{th:algo32}
Algorithm \emph{BHS-Torus-32} correctly solves the BHS problem in any oriented torus with exactly three agents carrying two tokens each.
\end{theorem}

\begin{proof}
We proved that if an agent reports it found the black hole, then it is
always correct. 

First, we know that if two or three agents start in the horizontal
ring containing the black hole, then one agent locates the black hole.

Otherwise, the agents execute \NR on consecutive rings until they see
some tokens belonging to dead agents. Thus, at least one agent will
die entering the black hole. Then, either a second agent locates the
black hole, or it also dies entering the black hole, leaving a
black hole configuration. In this last case, we know the third agent
will not be killed and will eventually locates the black hole. 
\end{proof}

\section{Conclusions}

We showed that at least three agents are needed to solve BHS in oriented torus networks and these three agents must carry at least two movable tokens each for marking the nodes.
The algorithm BHS-Torus-32 uses the smallest possible team of agents (i.e., 3) carrying the minimum number of tokens (i.e., 2) and thus, it is optimal in terms of resource requirements. However, on the downside this algorithm works only for $k=3$ agents. In combination with algorithm BHS-Torus-42 (which solves the problem for any $k\geq 4$ agents carrying $2$ tokens each), these algorithms can solve the black hole search problem for any $k\geq 3$, if the value of $k$ is known. Unfortunately, algorithms BHS-Torus-32 and BHS-Torus-42 cannot be combined to give an algorithm for solving the BHS problem for any $k\geq 3$ agents without the knowledge of $k$: Algorithm BHS-Torus-32 for $3$ agents will not correctly locate the black hole if the agents are more than $3$, while in the algorithm BHS-Torus-42, $3$ of the agents may fall into the black hole. 
Hence, whether the problem can be solved for $k\geq 3$ agents equipped with $2$ tokens, without any knowledge of $k$, remains an interesting (and we believe challenging) open question. Another interesting open problem is to determine the minimum size of a team of agents carrying one token each, that can solve the BHS problem. Note that the impossibility result for three agents carrying one token each, does not immediately generalize to the case of $4$ or more agents, as in those cases, we cannot exclude the possibility that two surviving agents manage to meet.

It is interesting to compare our results with the situation in a synchronous, oriented, anonymous ring, which can be seen as a one dimensional torus (\cite{cdlm11}): The minimum trade-offs between the number of agents and the number of tokens, in this case, are $4$ agents with $2$ unmovable tokens or $3$ agents with $1$ movable token each. Additionally, in an \emph{unoriented} ring the minimum trade-offs are $5$ agents with $2$ unmovable tokens or $3$ agents with $1$ movable token each whereas the situation in an unoriented torus has not been studied. Hence another open problem is solving the BHS problem in a $d$-dimensional torus, $d>3$, as well as in other network topologies.



\bibliographystyle{abbrv}
\bibliography{biblio-tot}

\end{document}

%% file: motifs.pdf_t
\begin{picture}(0,0)%
\includegraphics{motifs.pdf}%
\end{picture}%
\setlength{\unitlength}{3158sp}%
\begingroup\makeatletter\ifx\SetFigFont\undefined%
\gdef\SetFigFont#1#2#3#4#5{%
  \reset@font\fontsize{#1}{#2pt}%
  \fontfamily{#3}\fontseries{#4}\fontshape{#5}%
  \selectfont}%
\fi\endgroup%
\begin{picture}(10244,10244)(3579,-12083)
\put(4276,-7636){\makebox(0,0)[lb]{\smash{{\SetFigFont{8}{9.6}{\familydefault}{\mddefault}{\updefault}{\color[rgb]{0,0,0}0}%
}}}}
\put(4276,-8536){\makebox(0,0)[lb]{\smash{{\SetFigFont{8}{9.6}{\familydefault}{\mddefault}{\updefault}{\color[rgb]{0,0,0}1}%
}}}}
\put(4276,-8236){\makebox(0,0)[lb]{\smash{{\SetFigFont{8}{9.6}{\familydefault}{\mddefault}{\updefault}{\color[rgb]{0,.56,0}$[0-]$}%
}}}}
\put(7876,-7636){\makebox(0,0)[lb]{\smash{{\SetFigFont{8}{9.6}{\familydefault}{\mddefault}{\updefault}{\color[rgb]{0,0,0}0,2,4,6}%
}}}}
\put(8776,-7636){\makebox(0,0)[lb]{\smash{{\SetFigFont{8}{9.6}{\familydefault}{\mddefault}{\updefault}{\color[rgb]{0,0,0}3,7}%
}}}}
\put(7876,-8536){\makebox(0,0)[lb]{\smash{{\SetFigFont{8}{9.6}{\familydefault}{\mddefault}{\updefault}{\color[rgb]{0,0,0}1,5}%
}}}}
\put(8776,-8536){\makebox(0,0)[lb]{\smash{{\SetFigFont{8}{9.6}{\familydefault}{\mddefault}{\updefault}{\color[rgb]{0,0,0}8}%
}}}}
\put(7876,-8236){\makebox(0,0)[lb]{\smash{{\SetFigFont{8}{9.6}{\familydefault}{\mddefault}{\updefault}{\color[rgb]{0,.56,0}$[0-2]$}%
}}}}
\put(8776,-8236){\makebox(0,0)[lb]{\smash{{\SetFigFont{8}{9.6}{\familydefault}{\mddefault}{\updefault}{\color[rgb]{0,.56,0}$[3-]$}%
}}}}
\put(11476,-7636){\makebox(0,0)[lb]{\smash{{\SetFigFont{8}{9.6}{\familydefault}{\mddefault}{\updefault}{\color[rgb]{0,0,0}0,2}%
}}}}
\put(12376,-7636){\makebox(0,0)[lb]{\smash{{\SetFigFont{8}{9.6}{\familydefault}{\mddefault}{\updefault}{\color[rgb]{0,0,0}3,4,6}%
}}}}
\put(13276,-7636){\makebox(0,0)[lb]{\smash{{\SetFigFont{8}{9.6}{\familydefault}{\mddefault}{\updefault}{\color[rgb]{0,0,0}5}%
}}}}
\put(11476,-8536){\makebox(0,0)[lb]{\smash{{\SetFigFont{8}{9.6}{\familydefault}{\mddefault}{\updefault}{\color[rgb]{0,0,0}1,5}%
}}}}
\put(12376,-8536){\makebox(0,0)[lb]{\smash{{\SetFigFont{8}{9.6}{\familydefault}{\mddefault}{\updefault}{\color[rgb]{0,0,0}7}%
}}}}
\put(11476,-8236){\makebox(0,0)[lb]{\smash{{\SetFigFont{8}{9.6}{\familydefault}{\mddefault}{\updefault}{\color[rgb]{0,.56,0}$[0-2]$}%
}}}}
\put(12376,-8236){\makebox(0,0)[lb]{\smash{{\SetFigFont{8}{9.6}{\familydefault}{\mddefault}{\updefault}{\color[rgb]{0,.56,0}$[3-]$}%
}}}}
\put(11476,-10336){\makebox(0,0)[lb]{\smash{{\SetFigFont{8}{9.6}{\familydefault}{\mddefault}{\updefault}{\color[rgb]{0,0,0}0,2}%
}}}}
\put(7876,-2236){\makebox(0,0)[lb]{\smash{{\SetFigFont{8}{9.6}{\familydefault}{\mddefault}{\updefault}{\color[rgb]{0,0,0}0,2,4}%
}}}}
\put(7876,-3136){\makebox(0,0)[lb]{\smash{{\SetFigFont{8}{9.6}{\familydefault}{\mddefault}{\updefault}{\color[rgb]{0,0,0}1,5}%
}}}}
\put(8776,-2236){\makebox(0,0)[lb]{\smash{{\SetFigFont{8}{9.6}{\familydefault}{\mddefault}{\updefault}{\color[rgb]{0,0,0}3}%
}}}}
\put(7876,-2836){\makebox(0,0)[lb]{\smash{{\SetFigFont{8}{9.6}{\familydefault}{\mddefault}{\updefault}{\color[rgb]{0,.56,0}$[0-2]$}%
}}}}
\put(8776,-2836){\makebox(0,0)[lb]{\smash{{\SetFigFont{8}{9.6}{\familydefault}{\mddefault}{\updefault}{\color[rgb]{0,.56,0}$[3-]$}%
}}}}
\put(11476,-2236){\makebox(0,0)[lb]{\smash{{\SetFigFont{8}{9.6}{\familydefault}{\mddefault}{\updefault}{\color[rgb]{0,0,0}0,2,4}%
}}}}
\put(12376,-2236){\makebox(0,0)[lb]{\smash{{\SetFigFont{8}{9.6}{\familydefault}{\mddefault}{\updefault}{\color[rgb]{0,0,0}3}%
}}}}
\put(11476,-3136){\makebox(0,0)[lb]{\smash{{\SetFigFont{8}{9.6}{\familydefault}{\mddefault}{\updefault}{\color[rgb]{0,0,0}1,5}%
}}}}
\put(12376,-3136){\makebox(0,0)[lb]{\smash{{\SetFigFont{8}{9.6}{\familydefault}{\mddefault}{\updefault}{\color[rgb]{0,0,0}6}%
}}}}
\put(11476,-2836){\makebox(0,0)[lb]{\smash{{\SetFigFont{8}{9.6}{\familydefault}{\mddefault}{\updefault}{\color[rgb]{0,.56,0}$[0-2]$}%
}}}}
\put(12376,-2836){\makebox(0,0)[lb]{\smash{{\SetFigFont{8}{9.6}{\familydefault}{\mddefault}{\updefault}{\color[rgb]{0,.56,0}$[3-]$}%
}}}}
\put(4276,-5536){\makebox(0,0)[lb]{\smash{{\SetFigFont{8}{9.6}{\familydefault}{\mddefault}{\updefault}{\color[rgb]{0,.56,0}$[0-2]$}%
}}}}
\put(5176,-5536){\makebox(0,0)[lb]{\smash{{\SetFigFont{8}{9.6}{\familydefault}{\mddefault}{\updefault}{\color[rgb]{0,.56,0}$[3-9]$}%
}}}}
\put(4276,-4936){\makebox(0,0)[lb]{\smash{{\SetFigFont{8}{9.6}{\familydefault}{\mddefault}{\updefault}{\color[rgb]{0,0,0}0,2,4,6,10}%
}}}}
\put(4276,-5836){\makebox(0,0)[lb]{\smash{{\SetFigFont{8}{9.6}{\familydefault}{\mddefault}{\updefault}{\color[rgb]{0,0,0}1,5,11}%
}}}}
\put(5176,-4936){\makebox(0,0)[lb]{\smash{{\SetFigFont{8}{9.6}{\familydefault}{\mddefault}{\updefault}{\color[rgb]{0,0,0}3,7,9}%
}}}}
\put(5176,-5836){\makebox(0,0)[lb]{\smash{{\SetFigFont{8}{9.6}{\familydefault}{\mddefault}{\updefault}{\color[rgb]{0,0,0}8}%
}}}}
\put(7876,-5836){\makebox(0,0)[lb]{\smash{{\SetFigFont{8}{9.6}{\familydefault}{\mddefault}{\updefault}{\color[rgb]{0,0,0}1}%
}}}}
\put(7876,-4936){\makebox(0,0)[lb]{\smash{{\SetFigFont{8}{9.6}{\familydefault}{\mddefault}{\updefault}{\color[rgb]{0,0,0}0,2}%
}}}}
\put(8776,-4936){\makebox(0,0)[lb]{\smash{{\SetFigFont{8}{9.6}{\familydefault}{\mddefault}{\updefault}{\color[rgb]{0,0,0}3}%
}}}}
\put(7876,-5536){\makebox(0,0)[lb]{\smash{{\SetFigFont{8}{9.6}{\familydefault}{\mddefault}{\updefault}{\color[rgb]{0,.56,0}$[0-2]$}%
}}}}
\put(7876,-11236){\makebox(0,0)[lb]{\smash{{\SetFigFont{8}{9.6}{\familydefault}{\mddefault}{\updefault}{\color[rgb]{0,0,0}1}%
}}}}
\put(11476,-10936){\makebox(0,0)[lb]{\smash{{\SetFigFont{8}{9.6}{\familydefault}{\mddefault}{\updefault}{\color[rgb]{0,.56,0}$[0-2]$}%
}}}}
\put(12376,-10936){\makebox(0,0)[lb]{\smash{{\SetFigFont{8}{9.6}{\familydefault}{\mddefault}{\updefault}{\color[rgb]{0,.56,0}$[3-8]$}%
}}}}
\put(13276,-10936){\makebox(0,0)[lb]{\smash{{\SetFigFont{8}{9.6}{\familydefault}{\mddefault}{\updefault}{\color[rgb]{0,.56,0}$[9-11]$}%
}}}}
\put(12376,-10336){\makebox(0,0)[lb]{\smash{{\SetFigFont{8}{9.6}{\familydefault}{\mddefault}{\updefault}{\color[rgb]{0,0,0}3,4,6,8}%
}}}}
\put(13276,-10336){\makebox(0,0)[lb]{\smash{{\SetFigFont{8}{9.6}{\familydefault}{\mddefault}{\updefault}{\color[rgb]{0,0,0}5,9,11}%
}}}}
\put(11476,-11236){\makebox(0,0)[lb]{\smash{{\SetFigFont{8}{9.6}{\familydefault}{\mddefault}{\updefault}{\color[rgb]{0,0,0}1}%
}}}}
\put(12376,-11236){\makebox(0,0)[lb]{\smash{{\SetFigFont{8}{9.6}{\familydefault}{\mddefault}{\updefault}{\color[rgb]{0,0,0}7}%
}}}}
\put(13276,-11236){\makebox(0,0)[lb]{\smash{{\SetFigFont{8}{9.6}{\familydefault}{\mddefault}{\updefault}{\color[rgb]{0,0,0}10,12}%
}}}}
\put(4276,-2236){\makebox(0,0)[lb]{\smash{{\SetFigFont{8}{9.6}{\familydefault}{\mddefault}{\updefault}{\color[rgb]{0,0,0}0,2,4,6}%
}}}}
\put(4276,-2836){\makebox(0,0)[lb]{\smash{{\SetFigFont{8}{9.6}{\familydefault}{\mddefault}{\updefault}{\color[rgb]{0,.56,0}$[0-2]$}%
}}}}
\put(5176,-2836){\makebox(0,0)[lb]{\smash{{\SetFigFont{8}{9.6}{\familydefault}{\mddefault}{\updefault}{\color[rgb]{0,.56,0}$[3-]$}%
}}}}
\put(4276,-3136){\makebox(0,0)[lb]{\smash{{\SetFigFont{8}{9.6}{\familydefault}{\mddefault}{\updefault}{\color[rgb]{0,0,0}1,5}%
}}}}
\put(5176,-3136){\makebox(0,0)[lb]{\smash{{\SetFigFont{8}{9.6}{\familydefault}{\mddefault}{\updefault}{\color[rgb]{0,0,0}8}%
}}}}
\put(5176,-2236){\makebox(0,0)[lb]{\smash{{\SetFigFont{8}{9.6}{\familydefault}{\mddefault}{\updefault}{\color[rgb]{0,0,0}3,7,9}%
}}}}
\put(5176,-10336){\makebox(0,0)[lb]{\smash{{\SetFigFont{8}{9.6}{\familydefault}{\mddefault}{\updefault}{\color[rgb]{0,0,0}3,4,6,8}%
}}}}
\put(4276,-10336){\makebox(0,0)[lb]{\smash{{\SetFigFont{8}{9.6}{\familydefault}{\mddefault}{\updefault}{\color[rgb]{0,0,0}0,2}%
}}}}
\put(4276,-10936){\makebox(0,0)[lb]{\smash{{\SetFigFont{8}{9.6}{\familydefault}{\mddefault}{\updefault}{\color[rgb]{0,.56,0}$[0-2]$}%
}}}}
\put(5176,-10936){\makebox(0,0)[lb]{\smash{{\SetFigFont{8}{9.6}{\familydefault}{\mddefault}{\updefault}{\color[rgb]{0,.56,0}$[3-8]$}%
}}}}
\put(6076,-10936){\makebox(0,0)[lb]{\smash{{\SetFigFont{8}{9.6}{\familydefault}{\mddefault}{\updefault}{\color[rgb]{0,.56,0}$[9-]$}%
}}}}
\put(4276,-11236){\makebox(0,0)[lb]{\smash{{\SetFigFont{8}{9.6}{\familydefault}{\mddefault}{\updefault}{\color[rgb]{0,0,0}1}%
}}}}
\put(5176,-11236){\makebox(0,0)[lb]{\smash{{\SetFigFont{8}{9.6}{\familydefault}{\mddefault}{\updefault}{\color[rgb]{0,0,0}7}%
}}}}
\put(6076,-10336){\makebox(0,0)[lb]{\smash{{\SetFigFont{8}{9.6}{\familydefault}{\mddefault}{\updefault}{\color[rgb]{0,0,0}5,9}%
}}}}
\put(7876,-10936){\makebox(0,0)[lb]{\smash{{\SetFigFont{8}{9.6}{\familydefault}{\mddefault}{\updefault}{\color[rgb]{0,.56,0}$[0-2]$}%
}}}}
\put(8776,-10936){\makebox(0,0)[lb]{\smash{{\SetFigFont{8}{9.6}{\familydefault}{\mddefault}{\updefault}{\color[rgb]{0,.56,0}$[3-7]$}%
}}}}
\put(8776,-10336){\makebox(0,0)[lb]{\smash{{\SetFigFont{8}{9.6}{\familydefault}{\mddefault}{\updefault}{\color[rgb]{0,0,0}3,4,7}%
}}}}
\put(9676,-10336){\makebox(0,0)[lb]{\smash{{\SetFigFont{8}{9.6}{\familydefault}{\mddefault}{\updefault}{\color[rgb]{0,0,0}5,6,8}%
}}}}
\put(7876,-10336){\makebox(0,0)[lb]{\smash{{\SetFigFont{8}{9.6}{\familydefault}{\mddefault}{\updefault}{\color[rgb]{0,0,0}0,2}%
}}}}
\put(13201,-10636){\makebox(0,0)[b]{\smash{{\SetFigFont{8}{9.6}{\familydefault}{\mddefault}{\updefault}{\color[rgb]{1,0,0}0}%
}}}}
\put(6001,-10636){\makebox(0,0)[b]{\smash{{\SetFigFont{8}{9.6}{\familydefault}{\mddefault}{\updefault}{\color[rgb]{1,0,0}0}%
}}}}
\put(5101,-11536){\makebox(0,0)[b]{\smash{{\SetFigFont{8}{9.6}{\familydefault}{\mddefault}{\updefault}{\color[rgb]{1,0,0}0}%
}}}}
\put(4201,-11536){\makebox(0,0)[b]{\smash{{\SetFigFont{8}{9.6}{\familydefault}{\mddefault}{\updefault}{\color[rgb]{1,0,0}1}%
}}}}
\put(13201,-11536){\makebox(0,0)[b]{\smash{{\SetFigFont{8}{9.6}{\familydefault}{\mddefault}{\updefault}{\color[rgb]{1,0,0}$\leq$2}%
}}}}
\put(11401,-3436){\makebox(0,0)[b]{\smash{{\SetFigFont{8}{9.6}{\familydefault}{\mddefault}{\updefault}{\color[rgb]{1,0,0}2/0}%
}}}}
\put(12301,-2536){\makebox(0,0)[b]{\smash{{\SetFigFont{8}{9.6}{\familydefault}{\mddefault}{\updefault}{\color[rgb]{1,0,0}0}%
}}}}
\put(8701,-2536){\makebox(0,0)[b]{\smash{{\SetFigFont{8}{9.6}{\familydefault}{\mddefault}{\updefault}{\color[rgb]{1,0,0}0}%
}}}}
\put(7801,-3436){\makebox(0,0)[b]{\smash{{\SetFigFont{8}{9.6}{\familydefault}{\mddefault}{\updefault}{\color[rgb]{1,0,0}0/2}%
}}}}
\put(5101,-3436){\makebox(0,0)[b]{\smash{{\SetFigFont{8}{9.6}{\familydefault}{\mddefault}{\updefault}{\color[rgb]{1,0,0}$\leq$2}%
}}}}
\put(5101,-2536){\makebox(0,0)[b]{\smash{{\SetFigFont{8}{9.6}{\familydefault}{\mddefault}{\updefault}{\color[rgb]{1,0,0}0}%
}}}}
\put(4201,-3436){\makebox(0,0)[b]{\smash{{\SetFigFont{8}{9.6}{\familydefault}{\mddefault}{\updefault}{\color[rgb]{1,0,0}0}%
}}}}
\put(12301,-3436){\makebox(0,0)[b]{\smash{{\SetFigFont{8}{9.6}{\familydefault}{\mddefault}{\updefault}{\color[rgb]{1,0,0}$\leq$2}%
}}}}
\put(12301,-11536){\makebox(0,0)[b]{\smash{{\SetFigFont{8}{9.6}{\familydefault}{\mddefault}{\updefault}{\color[rgb]{1,0,0}$>$0}%
}}}}
\put(7801,-6136){\makebox(0,0)[b]{\smash{{\SetFigFont{8}{9.6}{\familydefault}{\mddefault}{\updefault}{\color[rgb]{1,0,0}$\leq$2}%
}}}}
\put(8701,-5236){\makebox(0,0)[b]{\smash{{\SetFigFont{8}{9.6}{\familydefault}{\mddefault}{\updefault}{\color[rgb]{1,0,0}2}%
}}}}
\put(5101,-6136){\makebox(0,0)[b]{\smash{{\SetFigFont{8}{9.6}{\familydefault}{\mddefault}{\updefault}{\color[rgb]{1,0,0}$\leq$2}%
}}}}
\put(4201,-6136){\makebox(0,0)[b]{\smash{{\SetFigFont{8}{9.6}{\familydefault}{\mddefault}{\updefault}{\color[rgb]{1,0,0}2}%
}}}}
\put(5101,-5236){\makebox(0,0)[b]{\smash{{\SetFigFont{8}{9.6}{\familydefault}{\mddefault}{\updefault}{\color[rgb]{1,0,0}0}%
}}}}
\put(11401,-8836){\makebox(0,0)[b]{\smash{{\SetFigFont{8}{9.6}{\familydefault}{\mddefault}{\updefault}{\color[rgb]{1,0,0}1}%
}}}}
\put(12301,-7936){\makebox(0,0)[b]{\smash{{\SetFigFont{8}{9.6}{\familydefault}{\mddefault}{\updefault}{\color[rgb]{1,0,0}0}%
}}}}
\put(13201,-7936){\makebox(0,0)[b]{\smash{{\SetFigFont{8}{9.6}{\familydefault}{\mddefault}{\updefault}{\color[rgb]{1,0,0}0}%
}}}}
\put(7801,-8836){\makebox(0,0)[b]{\smash{{\SetFigFont{8}{9.6}{\familydefault}{\mddefault}{\updefault}{\color[rgb]{1,0,0}2}%
}}}}
\put(8701,-7936){\makebox(0,0)[b]{\smash{{\SetFigFont{8}{9.6}{\familydefault}{\mddefault}{\updefault}{\color[rgb]{1,0,0}0}%
}}}}
\put(11401,-11536){\makebox(0,0)[b]{\smash{{\SetFigFont{8}{9.6}{\familydefault}{\mddefault}{\updefault}{\color[rgb]{1,0,0}1}%
}}}}
\put(12301,-10636){\makebox(0,0)[b]{\smash{{\SetFigFont{8}{9.6}{\familydefault}{\mddefault}{\updefault}{\color[rgb]{1,0,0}0}%
}}}}
\put(8701,-10636){\makebox(0,0)[b]{\smash{{\SetFigFont{8}{9.6}{\familydefault}{\mddefault}{\updefault}{\color[rgb]{1,0,0}0}%
}}}}
\put(9601,-10636){\makebox(0,0)[b]{\smash{{\SetFigFont{8}{9.6}{\familydefault}{\mddefault}{\updefault}{\color[rgb]{1,0,0}2}%
}}}}
\put(7801,-11536){\makebox(0,0)[b]{\smash{{\SetFigFont{8}{9.6}{\familydefault}{\mddefault}{\updefault}{\color[rgb]{1,0,0}1}%
}}}}
\put(5101,-10636){\makebox(0,0)[b]{\smash{{\SetFigFont{8}{9.6}{\familydefault}{\mddefault}{\updefault}{\color[rgb]{1,0,0}0}%
}}}}
\end{picture}%